\documentclass[11pt,draftcls,onecolumn]{IEEEtran}

\usepackage{indentfirst,flushend}
\usepackage{amssymb,bm,mathrsfs,bbm,amscd,amsmath,amsthm,amsfonts,bbm}
\usepackage{dsfont,diagbox}
\usepackage{graphicx}
\usepackage{subfigure}
\usepackage{color}
\usepackage{cases}
\usepackage{stfloats}
\usepackage{cases}
\newtheorem{theorem}{Theorem}

\newtheorem{definition}{Definition}

\newtheorem{remark}{Remark}
\usepackage{multirow}
\usepackage[sort]{cite}
\usepackage{multicol}
\usepackage[nooneline,flushleft]{caption2}
\usepackage{clrscode}
\usepackage{algorithm}
\usepackage{algorithmic}
\usepackage{supertabular}
\allowdisplaybreaks[4]

\begin{document}

\title{Codebook Based Hybrid Precoding for Millimeter Wave Multiuser Systems}
\author{Shiwen~He,~\IEEEmembership{Member,~IEEE}, Jiaheng~Wang,~\IEEEmembership{Senior Member,~IEEE}, Yongming~Huang,~\IEEEmembership{Member,~IEEE},  ~Bj\"{o}rn~Ottersten,~\IEEEmembership{Fellow,~IEEE}, and Wei~Hong,~\IEEEmembership{Fellow,~IEEE}
\thanks{Manuscript received Jan. 06, 2017; revised Apr. 14, 2017;  accepted Jun. 24, 2017.}
\thanks{S. He is with the State Key Laboratory of Millimeter Waves, School of Information Science and Engineering,  Southeast University, Nanjing 210096, China. (Email: shiwenhe@seu.edu.cn).}
\thanks{J. Wang is with the National Mobile Communications Research Laboratory, School of Information Science and Engineering, Southeast University, Nanjing 210096, China. (Email: jhwang@seu.edu.cn).}
\thanks{Y. Huang (Corresponding author) is with the National Mobile Communications Research Laboratory, School of Information Science and Engineering, Southeast University, Nanjing 210096, China. (Email: huangym@seu.edu.cn).}
\thanks{B. Ottersten is with the Interdisciplinary Centre for Security Reliability and Trust (SnT), University of Luxembourg, Luxembourg, and also with the Royal Institute of Technology (KTH), Stockholm, Sweden. (e-mail: bjorn.ottersten@uni.lu; bjorn.ottersten@ee.kth.seg).}
\thanks{W. Hong is with the State Key Laboratory of Millimeter Waves, School of Information Science and Engineering, Southeast University, Nanjing 210096, China. (Email: weihong@seu.edu.cn). }
}

\maketitle
\vspace{-.6 in}

\begin{abstract}
In millimeter wave (mmWave) systems, antenna architecture limitations make it difficult to apply conventional fully digital precoding techniques but call for low cost analog radio-frequency (RF) and digital baseband hybrid precoding methods. This paper investigates joint RF-baseband hybrid precoding for the downlink of multiuser multi-antenna mmWave systems with a limited number of RF chains. Two performance measures, maximizing the spectral efficiency and the energy efficiency of the system, are considered. We propose a codebook based RF precoding design and obtain the channel state information via a beam sweep procedure. Via the codebook based design, the original system is transformed into a virtual multiuser downlink system with the RF chain constraint. Consequently, we are able to simplify the complicated hybrid precoding optimization problems to joint codeword selection and precoder design (JWSPD) problems. Then, we propose efficient methods to address the JWSPD problems and jointly optimize the RF and baseband precoders under the two performance measures. Finally, extensive numerical results are provided to validate the effectiveness of the proposed hybrid precoders.
\end{abstract}


\section*{\sc \uppercase\expandafter{\romannumeral1}. Introduction}

The proliferation of multimedia infotainment applications and high-end devices (e.g., smartphones, tablets, wearable devices, laptops, machine-to-machine communication devices) causes an explosive demand for high-rate data services. Future wireless communication systems face significant challenges in improving system capacity and guaranteeing users' quality of service (QoS) experiences~\cite{CommChen2012}. In the last few years, various physical layer enhancements, such as massive multiple-input multiple-output (MIMO)~\cite{TITBjor2014}, cooperation communication~\cite{CMagTao2012}, and network densification~\cite{JSACWang2007} have been proposed. Along with these technologies, there is a common agreement that exploiting higher frequency bands, such as the millimeter wave (mmWave) frequency bands, is a promising solution to increase network capacity for future wireless networks~\cite{JSTSPHeath2016}.

MmWave communication spans a wide frequency range from $30$ GHz to $300$ GHz and thus enjoys much wider bandwidth than today's cellular systems~\cite{ProcRangan2014}. However, mmWave signals experience more severe path loss, penetration loss, and rain fading compared with signals in sub-6 GHz frequency bands. For example, the free space path loss (FSPL) at $60$ GHz frequency bands is $35.6$ dB higher than that at 1 GHz~\cite{AccessRap2013,TAPRap2013,AccessMac2015}. Such a large FSPL must be compensated by the transceiver in mmWave communication systems. Fortunately, the very small wavelength of mmWave signals enables a large number of miniaturized antennas to be packed in small dimension, thus forming a large multi-antenna system potentially providing very large array gain. In conventional multi-antenna systems, each active transmit antenna is connected to a separate transmit radio frequency (RF) chain. Although physical antenna elements are cheap, transmit RF chains are not cheap. A large number of transmit RF chains not only increase the cost of RF circuits in terms of size and hardware but also consume additional energy in wireless communication systems~\cite{JSACCui2004}. Therefore, in practice, the number of RF chains is limited and much less than the number of antennas in mmWave systems.

For ease of implementation, fully analog beamforming was proposed in~\cite{TSPZhang2005,TCOMHur2013,TSPVen2010,JSACWang2009,CSTKutty2016}, where the phase of the signal sent by each antenna is manipulated via analog phase shifters. However, pure analog precoding (with only one RF chain) cannot provide multiplexing gains for transmitting parallel data streams. Hence, joint RF-baseband hybrid precoding, aiming to achieve both diversity and multiplexing gains, has attracted a great deal of interest in both academia and industry for mmWave communications~\cite{TWCEl2014,JSTSPAlk2014,TVTXue2016,CLLiang2014,TWCAlkha2015}. El Ayach \emph{et al} in~\cite{TWCEl2014} exploited the inherent sparsity of mmWave channels to design low-complexity hybrid precoders with perfect channel state information (CSI) at the receiver and partial CSI at the transmitter (CSIT). Alkhateeb \emph{et al} further investigated channel estimation for multi-path mmWave channels and tried to improve the performance of hybrid precoding using full CSIT~\cite{JSTSPAlk2014}. Note that the hybrid precoding designs in~\cite{TWCEl2014,JSTSPAlk2014} assume that either perfect or partial CSIT is available. In practice, while using partial CSIT may degrade system performance, perfect CSIT is often difficult to obtain in mmWave communication systems, especially when there are a large number of antennas. The RF-baseband hybrid precoders in~\cite{TWCEl2014,JSTSPAlk2014} were designed to obtain the spatial diversity or multiplexing gain for point-to-point mmWave communication systems. It is well known that multiuser communications can further provide multiuser diversity~\cite{TVTXue2016,CLLiang2014,TWCAlkha2015}. In~\cite{CLLiang2014}, the authors proposed a RF precoder for multiuser mmWave systems by matching the phase of the channel of each user also under the assumption of perfect CSIT. Later, a low-complexity codebook based RF-baseband hybrid precoder was proposed for a downlink multiuser mmWave system~\cite{TWCAlkha2015}. Note that both~\cite{CLLiang2014} and~\cite{TWCAlkha2015} assume  that the number of users equals the number of RF chains. In mmWave multiuser systems, it is very likely that the number of the served users per subcarrier will be less than that of RF chains. Therefore, it is necessary to study more flexible hybrid precoding designs for multiuser mmWave communication systems.

The existing RF-baseband hybrid precoding designs focus on improving the spectral efficiency of mmWave communication systems~\cite{TWCEl2014,JSTSPAlk2014,TVTXue2016,CLLiang2014,TWCAlkha2015}. On the other hand, accompanied by the growing energy demand and increasing energy price, the system energy efficiency (EE) becomes another critical performance measure for future wireless systems~\cite{TCOMHE2013,CLNguyen2015,TWCVen2015,TSPZap2016,TCOMLI2014}. In mmWave communication systems, although reducing the number of RF chains can save power consumption, the RF-baseband hybrid architecture requires additional power to operate the phase shifting network, the splitter, and the mixer at the transceiver~\cite{AccessRial2016}. Therefore, it is also necessary to investigate the RF-baseband hybrid precoding for improving the system EE. Recently, following the idea in~\cite{TWCEl2014}, an energy efficient hybrid precoding method was developed for 5G wireless communication systems with a large number of antennas and RF chains~\cite{JSACZi2016}. Differently, in this paper, we propose a codebook based hybrid precoding method that uses the effective CSIT to design the RF-baseband precoders.

In this paper, we study the RF-baseband hybrid precoding for the downlink of a multiuser multi-antenna mmWave communication system. The hybrid precoding design takes into account two hardware limitations: (\romannumeral1)  the analog phase shifters have constant modulus and a finite number of phase choices, and (\romannumeral2) the number of transmit RF chains is limited and less than the number of antennas. The design goal is to maximize the sum rate (SR) and the EE of the system. We introduce a codebook based RF precoding design along with a beam sweep procedure to reduce the complexity of the hybrid precoder and relieve the difficulty of obtaining CSIT. The contribution of this paper are summarized as follows.
\begin{itemize}
\item We investigate joint optimization of the RF-baseband precoders in multiuser mmWave systems under two common performance measures, i.e., maximizing the SR and the EE of the system.

\item Considering the practical limitation of phase shifters, we propose a codebook based RF precoder, whose columns (i.e., RF beamforming vectors) are specified by RF codewords, and then transform the original mmWave system into a virtual multiuser downlink multiple input single output (MISO) system.

\item We propose a beam sweep procedure to obtain effective CSIT with less signaling feedback by utilizing the beam-domain sparse property of mmWave channels.

\item Based on the codebook based design, we are able to simplify the original RF-baseband hybrid precoding optimization problems into joint codeword selection and precoding design (JWSPD) problems.

\item We propose an efficient method to address the JWSPD problem for maximizing the system SR.

\item We also develop an efficient method to address the more difficult JWSPD problem for maximizing the system EE.

\item Finally, extensive numerical results are provided to verify the effects of the proposed codebook based hybrid precoding design. It is shown that the proposed method outperforms the existing methods and achieves a satisfactory performance close to that of the fully digital precoder.
\end{itemize}

The remainder of this paper is organized as follows. The system model and optimization problem formulation are described in section \uppercase\expandafter{\romannumeral2}. Section \uppercase\expandafter{\romannumeral3} introduces a codebook based mmWave RF precoding design with beam sweep. An effective joint codewords selection and precoder design method is proposed for SRmax problem in section \uppercase\expandafter{\romannumeral4}. In section \uppercase\expandafter{\romannumeral5}, an effective joint codewords selection and precoder design method is developed for EEmax problem. In section \uppercase\expandafter{\romannumeral6}, numerical evaluations of these algorithms are carried out. Conclusions are finally drawn in section \uppercase\expandafter{\romannumeral7}.

\textbf{Notations}: Bold lowercase and uppercase letters represent column vectors and matrices, respectively. The superscripts $\left(\cdot\right)^{T}$, and $\left(\cdot\right)^{H}$ represent the transpose operator, and the conjugate transpose operator, respectively. $tr\left(\cdot\right)$, $\|\cdot\|_{2}$, $\left|\cdot\right|$, $\|\cdot\|_{\mathcal{F}}$, $\Re\left(\cdot\right)$ and $\Im\left(\cdot\right)$ denote the trace, the Euclidean norm, the absolute value (element-wise absolute if used with a matrix), Frobenius norm, the real and imaginary operators, respectively. $\bm{X}\geq\bm{Y}$ and $\bm{X}\leq\bm{Y}$ denote an element-wise inequality. $\bm{A}\succeq \bm{0}$ denotes matrix $\bm{A}$ is a semidefinite positive matrix. $\bm{1}_{N\times N}$ and  $\mathds{1}_{N}$ denote respectively $N\times N$ matrix with all one entries and $N\times 1$ all-one vector. $\bm{A}\left(m,n\right)$ represents the $\left(m{\rm th},n{\rm th}\right)$ element of matrix $\bm{A}$ and $diag\left(\bm{A}\right)$ stands for a column vector whose elements are the diagonal element of the matrix $\bm{A}$.  $\mathbb{R}$ and $\mathbb{C}$ are the real number field and the complex number field, respectively. $\log\left(\cdot\right)$ is the logarithm with base $e$. The function $floor\left(x\right)$ rounds the elements of $x$ to the nearest integers less than $x$. $\rm{mod}\left(,\right)$ is the modulo operation. $\upsilon_{max}^{\left(d\right)}\left(\bm{A}\right)$ is the set of right singular vectors corresponding to the $d$ largest singular values of matrix $\bm{A}$.

\section*{\sc \uppercase\expandafter{\romannumeral2}. Problem Statement}
\subsection*{A. System Model}
Consider the downlink of a mmWave multiuser multiple-input single-output (MISO) cellular system as shown in Fig.~\ref{SharedSystemModel}, where the BS is equipped with $M$ transmit antennas and $S$ RF chains and serves $K\leq S$ single-antenna users. Different from conventional multi-antenna communication systems, e.g.,~\cite{TCOMHE2013,CLNguyen2015,TSPZap2016}, where the numbers of antennas and RF chains are equal, in mmWave systems the number of antennas could be very large and it is expensive and impractical to install an RF chain for each antenna, so in practice we often have $S\leq M$. 
\begin{figure}[h]
\centering
\captionstyle{flushleft}
\onelinecaptionstrue
\includegraphics[width=0.8\columnwidth,keepaspectratio]{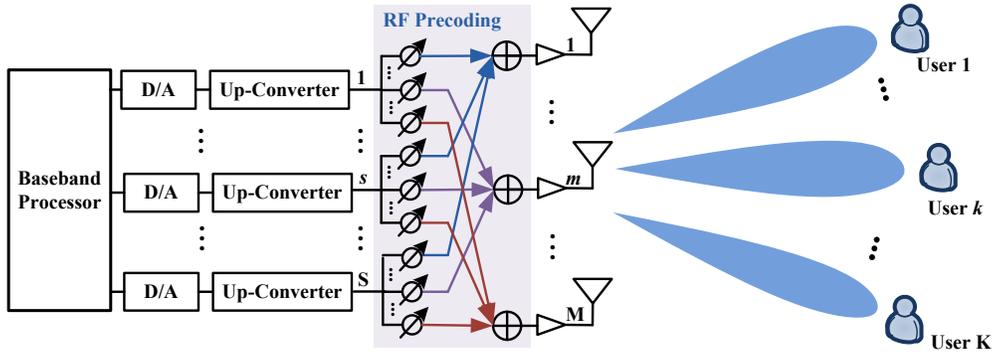}\\
\caption{Downlink mmWave system with hybrid RF-baseband precoding.}
\label{SharedSystemModel}
\end{figure}

To exploit the full potential of mmWave system with a limited number of RF chains, we consider an RF-baseband hybrid precoding design, in which the transmitted signal is precoded in both the (digital) baseband domain and the (analog) RF domain. Specifically, the system model can be expressed as
\begin{equation}\label{HybridMultiUserBeam01}
\bm{y}=\bm{H}\underline{\bm{F}}\underline{\bm{G}}\bm{s}+\bm{n},
\end{equation}
where $\bm{s}^{T}=\left[s_{1},\cdots,s_{K}\right]$ with $s_{k}\sim\mathcal{CN}\left(0,1\right)$ being the transmitted signal intended for the $k$th user, $\bm{y}=\left[y_{1},\cdots,y_{K}\right]^{T}$ with $y_{k}$ being the received signal of the $k$th user, $\bm{H}^{H}=\left[\bm{h}_{1},\cdots,\bm{h}_{K}\right]$ and $\bm{h}_{k}\in\mathds{C}^{M}$ contains the channel coefficients between the BS and the $k$th user, and $\bm{n}\sim\mathcal{CN}\left(\bm{0}, \sigma^{2}\bm{I}_{K}\right)$ is an additive white gaussian noise (AWGN) vector with independent identically distributed (i.i.d.) entries of zero mean and variance $\sigma^{2}$. In \eqref{HybridMultiUserBeam01}, $\underline{\bm{G}}\in\mathds{C}^{S\times K}$ is a baseband precoder that maps $\bm{s}$ to the $S$ RF chains, and $\underline{\bm{F}}\in\mathds{C}^{M\times S}$ is a RF precoder using analog circuitry, e.g., the analog phase shifting network. Due to the implementing limitation, the elements of $\underline{\bm{F}}$ are often required to have a constant modulus and only change their phases~\cite{TSPZhang2005}. Then, given the RF precoder $\underline{\bm{F}}$, the baseband precoder $\underline{\bm{G}}$, and the instantaneous CSI $\bm{h}_{k}, \forall k\in\mathcal{K}\triangleq\left\{1,2,\cdots,K\right\}$, the signal-to-interference-plus-noise ratio (SINR) of the $k$th user is
\begin{equation}\label{HybridMultiUserBeam02}
\underline{\mbox{SINR}}_{k}=\frac{\left|\bm{h}_{k}^{H}\underline{\bm{F}}\underline{\bm{g}}_{k}\right|^{2}}{\sum\limits_{l=1,l\neq k}^{K}\left|\bm{h}_{k}^{H}\underline{\bm{F}}\underline{\bm{g}_{l}}\right|^{2}+\sigma^{2}},
\end{equation}
where $\underline{\bm{g}}_{k}$ denotes the $k$th column of $\underline{\bm{G}}$.

\subsection*{B. Channel Model}

In this paper, the channel between the BS and each user is modeled as a narrowband clustered channel based on the extended Saleh-Valenzuela model that has been widely used in mmWave communications~\cite{IEEE802.11ad,IEEE802153c}. The channel coefficient vector $\bm{h}_{k}$ is assumed to be a sum of the contributions of $N_{cl}$ scattering clusters, each of which includes $N_{ray}$ propagation paths. Specifically, $\bm{h}_{k}$ can be written as~\cite{TWCEl2014}
\begin{equation}\label{HybridMultiUserBeam03}
\bm{h}_{k}=\sqrt{\frac{M}{N_{cl}N_{ray}}}\sum\limits_{m_{p}=1}^{N_{cl}}\sum\limits_{n_{p}=1}^{N_{ray}}\alpha_{m_{p},n_{p}}
\bm{a}\left(\phi_{m_{p},n_{p}},\theta_{m_{p},n_{p}}\right),
\end{equation}
where $\alpha_{m_{p},n_{p}}$ is a complex Gaussian random variable with zero mean and variance $\sigma_{\alpha,m_{p}}^{2}$ for the $n_{p}$th ray in the $m_{p}$th scattering cluster, and $\phi_{m_{p},n_{p}} \left(\theta_{m_{p},n_{p}}\right)$ is its azimuth (elevation) angle of departure (AoD). $\bm{a}\left(\phi_{m_{p},n_{p}},\theta_{m_{p},n_{p}}\right)$ is the normalized array response vector at an azimuth (elevation) angle of $\phi_{m_{p},n_{p}}\left(\theta_{m_{p},n_{p}}\right)$ and depends on the structure of the transmit antenna array only. The $N_{ray}$ azimuth and elevation angles of departure $\phi_{m_{p},n_{p}}$ and $\theta_{m_{p},n_{p}}$ within the cluster $m_{p}$ follow the Laplacian distributions with a uniformly-random mean cluster angle of $\phi_{m_{p}}$ and $\theta_{m_{p}}$, respectively, and a constant angular spread (standard deviation) of $\sigma_{\phi}$ and $\sigma_{\theta}$, respectively~\cite{TWCZhang2007}.

In particular, for an $M$-element uniform linear array (ULA), the array response vector is given by~\cite{AntennaTheory1997}
\begin{equation}\label{HybridMultiUserBeam04}
\bm{a}_{\rm{ULA}}\left(\phi\right)=\sqrt{\frac{1}{M}}
\left[1,e^{j\frac{2\pi}{\lambda_{s}}d\sin\left(\phi\right)},\cdots,
e^{j\left(M-1\right)\frac{2\pi}{\lambda_{s}}d\sin\left(\phi\right)}\right]^{T},
\end{equation}
where $\lambda_{s}$ is the signal wavelength, and $d$ is the inter-element spacing. For uniform planar array (UPA) in the $yz$-plane with $M_{1}$ and $M_{2}$ elements on the $y$ and $z$ axes respectively, the array response vector is given by~\cite{AntennaTheory1997}
\begin{equation}\label{HybridMultiUserBeam05}
\begin{split}
&\bm{a}_{\rm{UPA}}\left(\phi,\theta\right)=\sqrt{\frac{1}{M_{1}M_{2}}}\\
&\left[1,\cdots,e^{j\frac{2\pi}{\lambda_{s}}d\left(m_{p}\sin\left(\phi\right)\sin\left(\theta\right)+n_{p}\cos\left(\theta\right)\right)},\cdots,\atop
e^{j\frac{2\pi}{\lambda_{s}}d\left(\left(M_{1}-1\right)\sin\left(\phi\right)\sin\left(\theta\right)+\left(M_{2}-1\right)\cos\left(\theta\right)\right)}
\right]^{T},
\end{split}
\end{equation}
where the antenna array size is $M_{1}M_{2}$ and $0\leq m_{p}< M_{1}\left(0\leq n_{p}< M_{2}\right)$ is the $y\left(z\right)$ indices of an antenna element.

\subsection*{C. Problem Formulation}

The goal of this paper is to design proper RF-baseband hybrid precoders for the mmWave communication system. For this purpose, we consider two common performance measures: the system sum rate (SR) and the system energy efficiency (EE). The problem of maximizing the system SR (SRmax) is formulated as:
\begin{equation}\label{HybridMultiUserBeam06}
\begin{split}
&\max_{\underline{\bm{F}},\underline{\bm{G}}}~\sum\limits_{k=1}^{K} \underline{R}_{k},\\
s.t. &\underline{R}_{k}=\log\left(1+\underline{\mbox{SINR}}_{k}\right)\geq \gamma_{k}, \forall k\in\mathcal{K},\\
& \underline{\bm{F}}\in\mathcal{F}_{RF},~\left\|\underline{\bm{F}}\underline{\bm{G}}\right\|_{\mathcal{F}}^{2}\leq P.
\end{split}
\end{equation}
The problem of maximizing the system EE (EEmax) is formulated as:
\begin{subequations}\label{HybridMultiUserBeam07}
\begin{align}
&\max_{\underline{\bm{F}},\underline{\bm{G}}}~\frac{\sum\limits_{k=1}^{K} \underline{R}_{k}}
{\epsilon\sum\limits_{k=1}^{K}\left\|\underline{\bm{F}}\underline{\bm{g}}_{k}\right\|_{2}^2+Q_{dyn}},\label{HybridMultiUserBeam07a}\\
s.t. &\underline{R}_{k}=\log\left(1+\underline{\mbox{SINR}}_{k}\right)\geq \gamma_{k}, \forall k\in\mathcal{K},\label{HybridMultiUserBeam07b}\\
& \underline{\bm{F}}\in\mathcal{F}_{RF},~\left\|\underline{\bm{F}}\underline{\bm{G}}\right\|_{\mathcal{F}}^{2}\leq P. \label{HybridMultiUserBeam07c}
\end{align}
\end{subequations}
In the above two problems, $\mathcal{F}_{RF}$ is the set of feasible RF precoders, i.e., the set of $M\times S$ matrices with constant-modulus entries, $\gamma_{k}$ is the target rate of the $k$th user, $P$ is the maximum allowable transmit power, $\epsilon\geq 1$ is a constant which accounts for the inefficiency of the power amplifier (PA)~\cite{TWCKwanNg2012}. $Q_{dyn}$ is the dynamic power consumption, including the power radiation of all circuit blocks in each active RF chain and transmit antenna, given by
\begin{equation}\label{HybridMultiUserBeam08}
Q_{dyn}=\|\underline{\ddot{\bm{g}}}\|_{0}\left(P_{RFC}+MP_{PS}+P_{DAC}\right) +P_{sta},
\end{equation}
where $\ddot{\underline{\bm{g}}}=\left[\left\|\widetilde{\underline{\bm{g}}}_{1}\right\|_{2},\cdots,
\left\|\widetilde{\underline{\bm{g}}}_{S}\right\|_{2}\right]^{T}$ with $\widetilde{\underline{\bm{g}}}_{m}$ denoting the $m$th row of $\underline{\bm{G}}$, and the $\ell_{0}$-(quasi)norm $\|\ddot{\underline{\bm{g}}}\|_{0}$ is the number of nonzero entries of $\ddot{\underline{\bm{g}}}$, i.e., $\|\ddot{\underline{\bm{g}}}\|_{0}=\left|\left\{t: \left\|\widetilde{\underline{\bm{g}}}_{t}\right\|_{2}\neq 0\right\}\right|$. $P_{RFC}$, $P_{PS}$, and $P_{DAC}$ denote the the power consumption of the RF chain, the phase shifter (PS), and the digital-to-analog converter (DAC) at the transmitter, respectively. $P_{sta}=M\left(P_{PA}+P_{mixer}\right)+P_{BB}+P_{cool}$, where $P_{PA}$, $P_{mixer}$, $P_{BB}$, and $P_{cool}$  denote the power consumption of the PA, the mixer, the baseband signal processor, and the cooling system, respectively\footnote{The proposed framework in the paper can be readily extended to include the power consumption at the receivers.}.

The formulated problems~\eqref{HybridMultiUserBeam06} and~\eqref{HybridMultiUserBeam07} are challenging due to several difficulties, including the constant-modulus requirement of $\underline{\bm{F}}\in\mathcal{F}_{RF}$, the coupling between $\underline{\bm{G}}$ and $\underline{\bm{F}}$, the nonconvex nature of the user rates and the QoS constraints, and the fractional form of the objective (in problem \eqref{HybridMultiUserBeam07}). Another practical difficulty is the CSIT, which requires in general each user to estimate a large number of channels and feed them back to the BS. Throughout this paper, we assume that the set of user target rates is feasible. In the following, we will address these difficulties and propose efficient precoding designs.

\section*{\sc \uppercase\expandafter{\romannumeral3}. Codebook Based mmWave Precoding Design with Beam Sweeping}

In the mmWave system, the RF precoder is optimized in the analog domain and required to have a constant modulus. Unlike the digital baseband signal that can be precisely controlled, the RF signal is hard to manipulate and a precise shift for an arbitrary phase is prohibitively expensive in the analog domain. Therefore, in practice, each element of the RF precoder $\underline{\bm{F}}$ usually takes only several possible phase shifts, e.g., $8$ to $16$ choices ($3$ to $4$ bits), while the amplitude change is usually not possible~\cite{JSACWang2009,TSPZhang2005}. To facilitate the low complexity implementation of the phase shifter, the RF precoder is often selected from a predefined codebook, which contains a limited number of phase shifts with a constant amplitude.

An RF codebook can be represented by a matrix, where each column specifies a transmit pattern or an RF beamforming vector. In particular, let $\bm{F}\in\mathcal{F_{CB}}$ be an $M\times N$ predesigned codebook matrix, where $N$ is the number of codewords in the codebook $\bm{F}$, and $\mathcal{F_{CB}}$ denotes the space of all $M\times N$ constant-modulus RF precoding codewords. There are different RF codebooks, such as the general quantized beamforming codebooks and the beamsteering codebooks.

A $q$-bit resolution beam codebook for an $M$-element ULA is defined by a codebook matrix $\bm{F}$, where each column corresponds to a phase rotation of the antenna elements and generates a specific beam. A $q$-bit resolution codebook that achieves the uniform maximum gain in all directions with the optimal beamforming weight vectors is expressed as~\cite{CSTKutty2016}
\begin{equation}\label{HybridMultiUserBeam09}
\bm{F}\left(m,n\right)=\frac{1}{\sqrt{M}}j^{\frac{4\left(m-1\right)\left(n-1\right)-2N}{2^{q}}}, \forall m\in\mathcal{M}, \forall n\in\mathcal{N},
\end{equation}
where $j$ denotes the square root of $-1$, i.e.,  $j=\sqrt{-1}$, $\mathcal{M}=\left\{1,\cdots,M\right\}$, $\mathcal{N}=\left\{1,\cdots,N\right\}$.

The codebooks in IEEE 802.15.3c~\cite{IEEE802153c} and wireless personal area networks (WPAN) operating in $60$ GHz frequency band~\cite{IEEE802.11ad} are designed to simplify hardware implementation. The codebooks are generated with a $90$-degree phase resolution and without amplitude adjustment to reduce the power consumption. In this case, the $\left(m,n\right)$th element of the codebook $\bm{F}$ is given by~\eqref{HybridMultiUserBeam10}, $\forall m\in\mathcal{M}, \forall n\in\mathcal{N}$.
\begin{equation}\label{HybridMultiUserBeam10}
\bm{F}\left(m,n\right)=\frac{1}{\sqrt{M}}j^{floor\left(\frac{4\left(m-1\right)\left(\rm{mod}\left(\left(n-1\right)+\frac{N}{4},N\right)\right)}{N}\right)}.
\end{equation}
Note that when $M$ or $N$ is larger than $4$, the codebooks obtained from~\eqref{HybridMultiUserBeam10} result in the beam gain loss in some beam directions, due to the quantized phase shifts per antenna element with a limited $2$-bit codebook resolution.

In practice, discrete Fourier transform (DFT) codebooks are also widely used as they can achieve higher antenna gains at the beam directions than the codebooks in IEEE 802.15.3c. The entries of a DFT codebook are defined as
\begin{equation}\label{HybridMultiUserBeam11}
\bm{F}\left(m,n\right)\triangleq\frac{1}{\sqrt{M}}e^{-\frac{j2\pi\left(m-1\right)\left(n-1\right)}{M}}, \forall m\in\mathcal{M}, \forall n\in\mathcal{N}.
\end{equation}
The DFT codebooks generated in \eqref{HybridMultiUserBeam11} do not suffer any beam gain loss in the given beam directions for any $M$ and $N$. For mmWave systems, an efficient DFT codebook based MIMO beamforming training scheme was proposed in~\cite{ConfZhou2012} to estimate the antenna weight vectors (AWVs).

In Fig.~\ref{CodebookComparison}, we show the polar plots of array factor for two $3$-bit resolution codebooks using \eqref{HybridMultiUserBeam06} and \eqref{HybridMultiUserBeam11}, and a $2$-bit resolution codebook using \eqref{HybridMultiUserBeam10}. It can be observed that compared to the $2$-bit resolution codebook in IEEE 802.15.3c generated according to \eqref{HybridMultiUserBeam10}, the $3$-bit resolution beam codebook generated according to \eqref{HybridMultiUserBeam06} and the DFT codebook provide a better resolution and a symmetrical uniform maximum gain pattern with reduced side lobes.
\begin{figure}[h]
\centering
\captionstyle{flushleft}
\onelinecaptionstrue
\includegraphics[width=0.8\columnwidth,keepaspectratio]{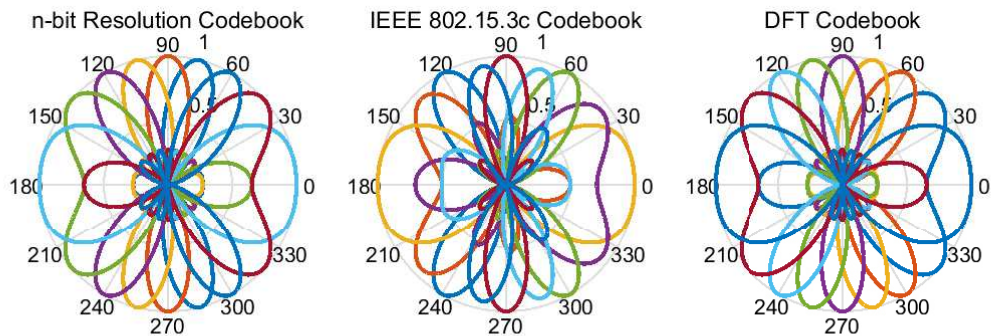}\\
\caption{Polar plots for array factor of $2$-bit and $3$-bit resolution codebooks with $8$ patterns, $M=4$, $N=8$.}
\label{CodebookComparison}
\end{figure}

Adopting an RF codebook dramatically redeuces the complexity of computing the RF precoder. Indeed, given an RF codebook $\bm{F}$, the optimization of the RF precoder $\underline{\bm{F}}$ in~\eqref{HybridMultiUserBeam06} and~\eqref{HybridMultiUserBeam07} is then equivalent to selecting $S$ codewords (columns) from the RF codebook (matrix) $\bm{F}$. Moreover, instead of obtaining directly the exact CSIT, we can obtain the equivalent CSIT via a beam-sweep procedure~\cite{IEEE802153c,IEEE802.11ad}. Specifically, during the beam-sweep procedure, the BS sends training packets from each direction defined in the RF codebook $\bm{F}$, and the users measure the received signal strength and estimate the effective channel across all directions. Then, each user provides the beam-sweep feedback to the BS, indicating the received signal strength and the effective channel of each direction, i.e., $\bm{h}_{k}^{H}\bm{f}_{n}$, where $\bm{f}_{n}$ is the $n$th codeword (column) of the RF codebook (matrix) $\bm{F}$. Such a beam-sweep procedure is shown in~Fig.~\ref{BeamformerTraining}.
\begin{figure}[h]
\centering
\captionstyle{flushleft}
\onelinecaptionstrue
\includegraphics[width=0.8\columnwidth,keepaspectratio]{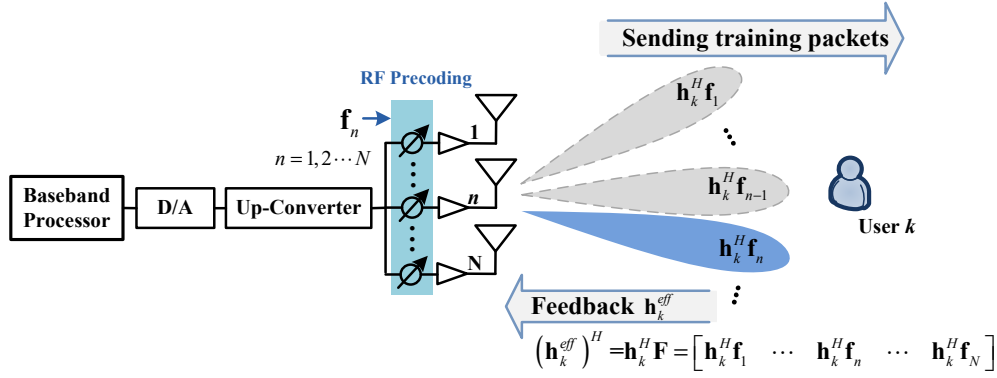}\\
\caption{Beam-sweep Procedure.}
\label{BeamformerTraining}
\end{figure}

\begin{remark}
\rm
Through the beam sweeping, the original system can be viewed as a virtual multiuser MISO downlink system, as illustrated in Fig.~\ref{EquivalentSystemModel}, where the BS is equipped with $N$ virtual antennas (i.e., codewords) and the channel coefficient between the BS and the $k$th user is $\bm{h}_{k}^{eff}=\bm{F}^{H}\bm{h}_{k}, \forall k\in\mathcal{K}$. It is well known that a mmWave channel equipped with a directional array usually admits a sparse property in the beam domain~\cite{TWCEl2014,JSTSPAlk2014}. That is, the effective channel may be near zero for most codewords $\bm{f}_{n}$ in the RF codebook $\bm{F}$. As a result, the effective channel coefficient vector $\bm{h}_{k}^{eff}$ is a sparse vector, implying that we only need to feedback a few nonzero effective channel coefficients to the BS. Therefore, by using a RF codebook along with the beam sweeping, the burden of obtaining CSIT in the mmWave system can be relieved.
\begin{figure}[h]
\centering
\captionstyle{flushleft}
\onelinecaptionstrue
\includegraphics[width=0.8\columnwidth,keepaspectratio]{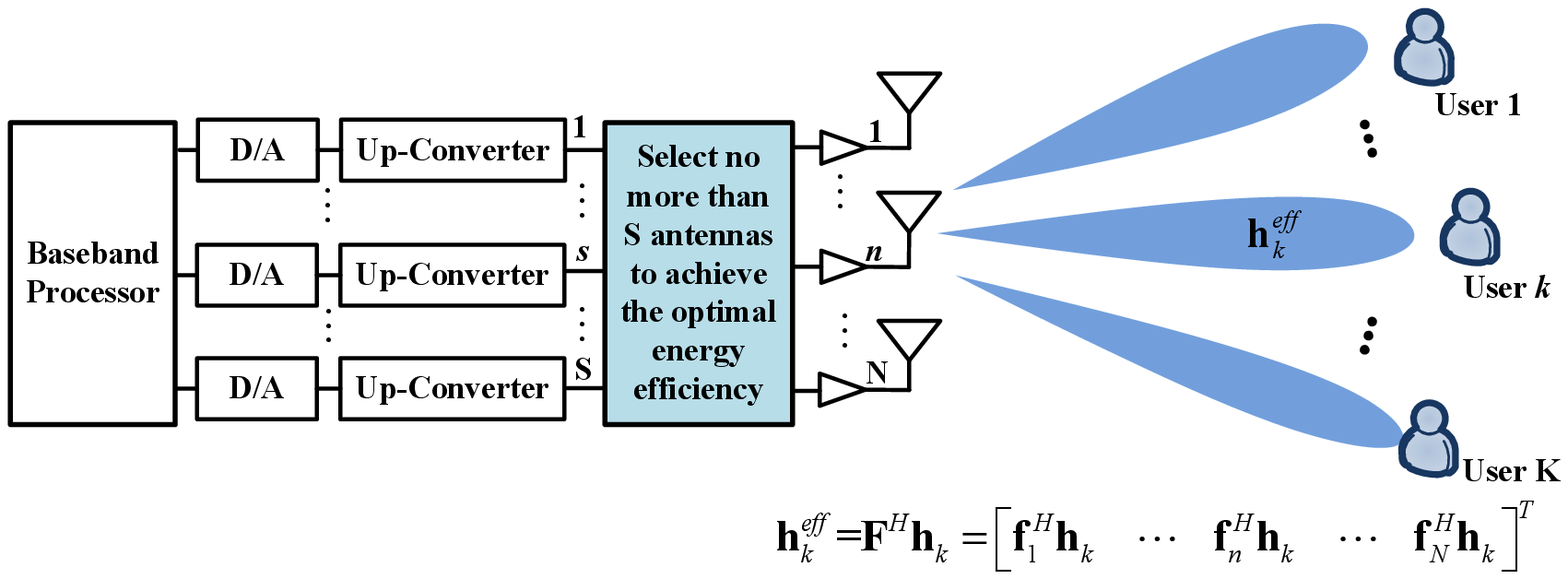}\\
\caption{Virtual Multiuser MISO Communication System.}
\label{EquivalentSystemModel}
\end{figure}
\end{remark}

Now, the hybrid precoding design becomes the joint optimization of the RF codeword selection and the baseband precoder. We show that this twofold task can be incorporated into the baseband precoder optimization. Specifically, instead of using the original $S\times K$ baseband precoder $\underline{\bm{G}}$, we introduce an expanded baseband precoder $\bm{G}\in\mathds{C}^{N\times K}$ with size of $N\times K$. Let $\widetilde{\bm{g}}_{m}$ denote the $m$th row of $\bm{G}$. Then, by multiplying the RF codebook $\bm{F}$ with $\bm{G}$, i.e., $\bm{F}\bm{G}$, the $m$th codeword in the RF codebook $\bm{F}$ is selected if and only if $\widetilde{\bm{g}}_{m}$ is nonzero or equivalently $\left\|\widetilde{\bm{g}}_{m}\right\|_{2}\neq 0$. Consequently, the original RF-baseband hybrid precoding design problem~(\ref{HybridMultiUserBeam06}) can be reformulated into the following joint codeword selection and precoder design (JWSPD) SRmax problem:
\begin{subequations}\label{HybridMultiUserBeam12}
\begin{align}
&\max_{\bm{G}}~\sum\limits_{k=1}^{K} R_{k},\label{HybridMultiUserBeam12a}\\
s.t.~&R_{k}=\log\left(1+\mbox{SINR}_{k}\right)\geq \gamma_{k}, \forall k\in\mathcal{K},\label{HybridMultiUserBeam12b}\\
&\sum\limits_{k=1}^{K}\left\|\bm{F}\bm{g}_{k}\right\|_{2}^2\leq P,~\|\ddot{\bm{g}}\|_{0}\leq S,\label{HybridMultiUserBeam12c}
\end{align}
\end{subequations}
where $\bm{g}_{k}$ denotes the $k$th column of $\bm{G}$, $\ddot{\bm{g}}=\left[\left\|\widetilde{\bm{g}}_{1}\right\|_{2},\cdots,\left\|\widetilde{\bm{g}}_{N}\right\|_{2}\right]^{T}$, the SINR of the $k$th user is given by
\begin{equation}\label{HybridMultiUserBeam13}
\mbox{SINR}_{k}=\frac{\left|\bm{h}_{k}^{H}\bm{F}\bm{g}_{k}\right|^{2}}{\sum\limits_{l=1,l\neq k}^{K}\left|\bm{h}_{k}^{H}\bm{F}\bm{g}_{l}\right|^{2}+\sigma^{2}}.
\end{equation}
In \eqref{HybridMultiUserBeam12}, the constraint $\|\ddot{\bm{g}}\|_{0}\leq S$ guarantees that the number of the selected codewords is no larger than the number of the available RF chains. Problem \eqref{HybridMultiUserBeam12} represents a sparse formulation of the baseband precoder design as $\ddot{\bm{g}}$ has up to $S\leq N$ nonzero elements. It also implies that the baseband precoder $\bm{G}$ is a sparse matrix.

Similarly, problem~\eqref{HybridMultiUserBeam07} can be reformulated into the following JWSPD EEmax problem:
\begin{subequations}\label{HybridMultiUserBeam14}
\begin{align}
&\max_{\bm{G}}~ \frac{\sum\limits_{k=1}^{K}R_{k}}
{\epsilon\sum\limits_{k=1}^{K}\left\|\bm{F}\bm{g}_{k}\right\|_{2}^2+P_{dyn}}\label{HybridMultiUserBeam14a}\\
s.t. &R_{k}=\log\left(1+\mbox{SINR}_{k}\right)\geq \gamma_{k}, \forall k\in\mathcal{K},\label{HybridMultiUserBeam14b}\\
&\sum\limits_{k=1}^{K}\left\|\bm{F}\bm{g}_{k}\right\|_{2}^2\leq P,~\|\ddot{\bm{g}}\|_{0}\leq S,\label{HybridMultiUserBeam14c}
\end{align}
\end{subequations}
where $P_{dyn}=\|\ddot{\bm{g}}\|_{0}\left(P_{RFC}+M P_{PS}+P_{DAC}\right)+P_{sta}$. Let $m_{l}$ be the row index of the $l$th nonzero row vector of $\bm{G}$ for $l=1,\cdots, \|\ddot{\bm{g}}\|_{0}$ with $m_{1}\leqslant\cdots\leqslant m_{\|\ddot{\bm{g}}\|_{0}}$. Without loss of generality, we can let the $l$th row vector of the baseband precoder be the $\widetilde{\bm{g}}_{m_{l}}$ and the $l$th phase shifter network steer vector be the $m_{l}$th codeword in the RF codebook $\bm{F}$ for the $l$th RF chain. Then, the remained $S-\|\ddot{\bm{g}}\|_{0}$ RF chains with the corresponding phase shifter networks can be turned off to save power.

So far, we have simplified the original RF-baseband hybrid precoding design into the JWSPD optimization problem. However, problems~\eqref{HybridMultiUserBeam12} and \eqref{HybridMultiUserBeam14}, although there is only one (matrix) variable $\bm{G}$, are still difficult, due to the nonconvex objective, the nonconvex QoS constraint, and the $\ell_{0}$-(quasi)norm constraint $\|\ddot{\bm{g}}\|_{0}\leqslant S$.

\section*{\sc \uppercase\expandafter{\romannumeral4}. Joint Codeword Selection and Precoder Optimization for SRmax Problem}

In this section, we consider first the JWSPD SRmax problem~\eqref{HybridMultiUserBeam12}, which, unfortunately, is NP-hard as a result of the nonconvex (sum rate) objective and the $\ell_{0}$-(quasi)norm constraint. Hence, finding its globally optimal solution requires prohibitive complexity, so in practice an efficient (probably suboptimal) solution is more preferred. In what follows, we will provide such an efficient solution.

\subsection*{A. Joint Codeword Selection and Precoder Design for SRmax problem}

To address the joint codeword selection and precoder design (JWSPD) in~\eqref{HybridMultiUserBeam12}, we first introduce some auxiliary variables $\alpha_{k}, \beta_{k}$, $\forall k\in\mathcal{K}$, $\tau$, $\kappa$, and $\chi$. Let $\log\left(1+\alpha_{k}\right)\geq\beta_{k}$ and $\mbox{SINR}_{k}\geq\alpha_{k}$, $\forall k\in\mathcal{K}$. After some basic operations, \eqref{HybridMultiUserBeam12} can be rewritten into the following equivalent form:
\begin{subequations}\label{HybridMultiUserBeam15}
\begin{align}
&\min_{\{\bm{g}_{k}, \alpha_{k}, \beta_{k}\}}~-\sum\limits_{k=1}^{K} \beta_{k}\label{HybridMultiUserBeam15a}\\
s.t.~&1+\alpha_{k}\geq e^{\beta_{k}},\forall k\in\mathcal{K},\label{HybridMultiUserBeam15b}\\
&\mbox{SINR}_{k}\geq\alpha_{k}, \mbox{SINR}_{k}\geq\overline{\gamma}_{k}, \forall k\in\mathcal{K},\label{HybridMultiUserBeam15c}\\
&\sum\limits_{k=1}^{K}\left\|\bm{F}\bm{g}_{k}\right\|_{2}^2\leq P,~\|\ddot{\bm{g}}\|_{0}\leq S\label{HybridMultiUserBeam15d},
\end{align}
\end{subequations}
where $\overline{\gamma}_{k}=e^{\gamma_{k}}-1$. It can be easily proven that the constraints~(\ref{HybridMultiUserBeam15b}) and $\mbox{SINR}_{k}\geq\alpha_{k}, \forall k$ shall be activated at the optimal solution~\cite{BoolBoyd2004}. The difficulty lies in (\ref{HybridMultiUserBeam15c}) and~\eqref{HybridMultiUserBeam15d}, as \eqref{HybridMultiUserBeam15c} and $\|\ddot{\bm{g}}\|_{0}\leq S$ are nonconvex constraints. To overcome these difficulties, we first move the constraint $\|\ddot{\bm{g}}\|_{0}\leq S$ into the objective as follows:
\begin{subequations}\label{HybridMultiUserBeam16}
\begin{align}
&\min_{\{\bm{g}_{k}, \alpha_{k}, \beta_{k}\}}~-\sum\limits_{k=1}^{K} \beta_{k}+\lambda\|\ddot{\bm{g}}\|_{0}\label{HybridMultiUserBeam16a}\\
s.t.~&1+\alpha_{k}\geq e^{\beta_{k}},\forall k\in\mathcal{K}, \sum\limits_{k=1}^{K}\left\|\bm{F}\bm{g}_{k}\right\|_{2}^2\leq P\label{HybridMultiUserBeam16b}\\
&\mbox{SINR}_{k}\geq\alpha_{k}, \mbox{SINR}_{k}\geq\overline{\gamma}_{k}, \forall k\in\mathcal{K},\label{HybridMultiUserBeam16c}
\end{align}
\end{subequations}
where $\lambda$ is a group-sparsity inducing regularization~\cite{TSPMehanna2013} to control the sparsity of the solution, i.e., the larger $\lambda$ is, the more sparse solution of (16) is.  Therefore, one can always choose a $\lambda$ large enough such that the constraint $\|\ddot{\bm{g}}\|_{0}\leq S$ is satisfied.

Then, we use the convex $\ell_{1,\infty}$-norm squared to approximate the nonconvex $\ell_{0}$-(quasi)norm\footnote{It is worth pointing out that the RF chain constraint $\|\ddot{\bm{g}}\|_{0}\leq S$ cannot be simply replaced by $\|\ddot{\bm{g}}\|_{p}\leq S$ with $p\geqslant 1$, since it is unknown whether $\ell_{0}$-norm $\geq$ $\ell_{p}$-norm or $\ell_{0}$-norm $<$ $\ell_{p}$-norm, which may result in a violation of the RF chain constraint.}. In this way, problem~\eqref{HybridMultiUserBeam16} is approximated as:
\begin{subequations}\label{HybridMultiUserBeam17}
\begin{align}
&\min_{\{\bm{g}_{k}, \alpha_{k}, \beta_{k}\}}~-\sum\limits_{k=1}^{K} \beta_{k}+\lambda\|\bm{G}\|_{1,\infty}^{2}\label{HybridMultiUserBeam17a}\\
s.t.~&1+\alpha_{k}\geq e^{\beta_{k}},\forall k\in\mathcal{K}, \sum\limits_{k=1}^{K}\left\|\bm{F}\bm{g}_{k}\right\|_{2}^2\leq P\label{HybridMultiUserBeam17b}\\
&\mbox{SINR}_{k}\geq\alpha_{k}, \mbox{SINR}_{k}\geq\overline{\gamma}_{k}, \forall k\in\mathcal{K},\label{HybridMultiUserBeam17c}
\end{align}
\end{subequations}
where $\|\bm{G}\|_{1,\infty}=\sum\limits_{n=1}^{N}\max\limits_{k}\left|\bm{g}_{k}\left(n\right)\right|$ is as the $\ell_{1,\infty}$-norm of the matrix $\bm{G}$. Note that $\|\bm{G}\|_{1,\infty}^2$ in~\eqref{HybridMultiUserBeam17a} can be rewritten as follows:
\begin{equation}\label{HybridMultiUserBeam18}
\begin{split}
\|\bm{G}\|_{1,\infty}^{2}&=\left(\sum\limits_{n=1}^{N}\max_{k}\left|\bm{g}_{k}\left(n\right)\right|\right)^2\\
&=\sum\limits_{n_{1}=1}^{N}\sum\limits_{n_{2}=1}^{N}\left(\left(\max_{k}\left|\bm{g}_{k}\left(n_{1}\right)\right|\right)
\left(\max_{k}\left|\bm{g}_{k}\left(n_{2}\right)\right|\right)\right)\\
&=\sum\limits_{n=1}^{N}\sum\limits_{m=1}^{N}
\max_{i,j\in\left\{1,\cdots,K\right\}}\left|\bm{X}_{i,j}\left(n,m\right)\right|,
\end{split}
\end{equation}
where $\bm{X}_{i,j}=\bm{g}_{i}\bm{g}_{j}^{H}$, $\forall i,j$. Note that $\bm{X}_{i,j}=\bm{g}_{i}\bm{g}_{j}^{H}$, $\forall i,j$ if and only if $\bm{X}_{i,j}\succeq \bm{0}$ and $\mbox{rank}\left(\bm{X}_{i,j}\right)=1$, $\forall i,j$. Thus, problem~\eqref{HybridMultiUserBeam17} can be relaxed to
\begin{subequations}\label{HybridMultiUserBeam19}
\begin{align}
&\min_{\{\bm{X}_{i,j}, \alpha_{k}, \beta_{k}\}}~-\sum\limits_{k=1}^{K} \beta_{k}+\lambda\|\bm{G}\|_{1,\infty}^{2},\label{HybridMultiUserBeam19a}\\
s.t.~&1+\alpha_{k}\geq e^{\beta_{k}},\forall k\in\mathcal{K}, \sum\limits_{k=1}^{K}\mbox{tr}\left(\widetilde{\bm{F}}\bm{X}_{k,k}\right)\leq P,\label{HybridMultiUserBeam19b}\\
&\mbox{SINR}_{k}\geq\alpha_{k}, \mbox{SINR}_{k}\geq\overline{\gamma}_{k},\bm{X}_{k,k}\succeq \bm{0}, \forall k\in\mathcal{K},\label{HybridMultiUserBeam19c}\\
&\mbox{rank}\left(\bm{X}_{i,j}\right)=1, \forall i,j,\label{HybridMultiUserBeam19d}
\end{align}
\end{subequations}
where $\widetilde{\bm{F}}=\bm{F}^{H}\bm{F}$, and $$\mbox{SINR}_{k}=\frac{\mbox{tr}\left(\bm{H}_{k}\bm{X}_{k,k}\right)}{\sum\limits_{l=1,l\neq k}^{K}\mbox{tr}\left(\bm{H}_{k}\bm{X}_{l,l}\right)+\sigma^{2}}$$
where $\bm{H}_{k}=\bm{F}^{H}\bm{h}_{k}\bm{h}_{k}^{H}\bm{F}$, $\forall k\in\mathcal{K}$. The relaxed problem~\eqref{HybridMultiUserBeam19} is still difficult as it is still nonconvex. Nevertheless, note that $\bm{X}_{i,j}$, $\forall i\neq j$ only appear in the objective~\eqref{HybridMultiUserBeam19a}, it is easy to have the following results which can help us simplify~\eqref{HybridMultiUserBeam19}.
\begin{theorem}\label{HybridMultiUserBeamL01}
Let $\{\breve{\bm{X}}_{i,j}, \breve{\alpha}_{k}, \breve{\beta}_{k}\}$ be the optimal solution of~(\ref{HybridMultiUserBeam19}), then the inequalities $\left|\breve{\bm{X}}_{i,j}\right|\leqslant\left|\breve{\bm{X}}_{i,i}\right|, \forall i\neq j$ hold.
\end{theorem}

For brevity, let $\bm{X}_{k}=\bm{X}_{k,k}, \forall k\in\mathcal{K}$ and define $\bm{Z}\left(n,m\right)=\max\limits_{{k}\in\mathcal{K}}\left|\bm{X}_{k}\left(n,m\right)\right|, \forall m,n$. Considering that the rank one constraint is nonconvex~\cite{TSPMehanna2013}, we obtain a tractable formulation form of problem~\eqref{HybridMultiUserBeam19} by dropping the nonconvex constraints $\mbox{rank}\left(\bm{X}_{k}\right)=1$, $\forall k\in\mathcal{K}$. According to Theorem~\ref{HybridMultiUserBeamL01}, problem (\ref{HybridMultiUserBeam19}) can be relaxed to:
\begin{subequations}\label{HybridMultiUserBeam21}
\begin{align}
&\min_{\{\bm{X}_{k}, \alpha_{k}, \beta_{k}\}, \bm{Z}}~-\sum\limits_{k=1}^{K} \beta_{k}+\lambda\mbox{tr}\left(\bm{1}_{N\times N}\bm{Z}\right)\label{HybridMultiUserBeam21a}\\
s.t.~&1+\alpha_{k}\geq e^{\beta_{k}},\forall k\in\mathcal{K}, \sum\limits_{k=1}^{K}\mbox{tr}\left(\widetilde{\bm{F}}\bm{X}_{k}\right)\leq P,\label{HybridMultiUserBeam21b}\\
&\mbox{SINR}_{k}\geq\alpha_{k}, \mbox{SINR}_{k}\geq\overline{\gamma}_{k}, \forall k\in\mathcal{K},\label{HybridMultiUserBeam21c}\\
&\bm{X}_{k}\succeq \bm{0}, \bm{Z}\geq\left|\bm{X}_{k}\right|, \forall k\in\mathcal{K}.
\end{align}
\end{subequations}
To address the nonconvex constraints~\eqref{HybridMultiUserBeam21c}, we transform it into the following problem~\eqref{HybridMultiUserBeam22}, at the top of this page, by introducing auxiliary variables $\psi_{k}, \phi_{k}$, $\forall k\in\mathcal{K}$, $\tau$, $\kappa$, and $\chi$,
\begin{subequations}\label{HybridMultiUserBeam22}
\begin{align}
&\min_{\{\bm{X}_{k}, \alpha_{k}, \beta_{k}, \psi_{k}, \phi_{k}\},\bm{Z}}~-\sum\limits_{k=1}^{K} \beta_{k}+\lambda\mbox{tr}\left(\bm{1}_{N\times N}\bm{Z}\right),\label{HybridMultiUserBeam22a}\\
s.t.~&\psi_{k}^{2}\leqslant\mbox{tr}\left(\bm{H}_{k}\bm{X}_{k}\right), \bm{X}_{k}\succeq \bm{0}, 1+\alpha_{k}\geq e^{\beta_{k}},  \forall k\in\mathcal{K}\label{HybridMultiUserBeam22b}\\
&\sum\limits_{l=1,l\neq k}^{K}\mbox{tr}\left(\bm{H}_{k}\bm{X}_{l}\right)+\sigma^{2}\leqslant\phi_{k}, \forall k\in\mathcal{K}, \sum\limits_{k=1}^{K}\mbox{tr}\left(\widetilde{\bm{F}}\bm{X}_{k}\right)\leq P,\label{HybridMultiUserBeam22c}\\
&\sum\limits_{l=1,l\neq k}^{K}\overline{\gamma}_{k}\mbox{tr}\left(\bm{H}_{k}\bm{X}_{l}\right)+\overline{\gamma}_{k}\sigma^{2}
\leqslant\mbox{tr}\left(\bm{H}_{k}\bm{X}_{k}\right),\frac{\psi_{k}^{2}}{\phi_{k}}\geqslant\alpha_{k},\forall k\in\mathcal{K},\label{HybridMultiUserBeam22d}\\
&\begin{bmatrix}
\bm{Z}\left(n,m\right)-\Re\left(\bm{X}_{k}\left(n,m\right)\right)&\Im\left(\bm{X}_{k}\left(n,m\right)\right)\\
\Im\left(\bm{X}_{k}\left(n,m\right)\right)&\bm{Z}\left(n,m\right)+\Re\left(\bm{X}_{k}\left(n,m\right)\right)
\end{bmatrix}\succeq\bm{0}, \forall k\in\mathcal{K},m,n.\label{HybridMultiUserBeam22e}
\end{align}
\end{subequations}
The difficulty of solving (\ref{HybridMultiUserBeam22}) lies in (\ref{HybridMultiUserBeam22d}), as the constraints $\frac{\psi_{k}^{2}}{\phi_{k}}\geqslant\alpha_{k}, \forall k$ are nonconvex. To overcome this difficulty, we exploit the SCA method~\cite{OperalMarks1977} to approximate the inequality $\frac{\psi_{k}^{2}}{\phi_{k}}\geqslant\alpha_{k}, \forall k$ by its convex low boundary as
\begin{equation}\label{HybridMultiUserBeam23}
\frac{\psi_{k}^{2}}{\phi_{k}}\geq\Phi_{k}^{\left(I\right)}\left(\psi_{k},\phi_{k}\right)
\triangleq 2\frac{\psi_{k}^{\left(I\right)}}{\phi_{k}^{\left(I\right)}}\psi_{k}
-\left(\frac{\psi_{k}^{\left(I\right)}}{\phi_{k}^{\left(I\right)}}\right)^{2}\phi_{k}, \forall k\in\mathcal{K},
\end{equation}
where the superscript $I$ denotes the $I$th iteration of the SCA method. Note that $\Phi_{k}^{\left(I\right)}\left(\psi_{k},\phi_{k}\right)$ is in fact the first order of $\frac{\psi_{k}^{2}}{\phi_{k}}$ around the point $\left(\psi_{k}^{\left(I\right)}, \phi_{k}^{\left(I\right)}\right)$. Thus, the approximate convex problem solved at iteration $I+1$ of~\eqref{HybridMultiUserBeam22} is given by:
\begin{subequations}\label{HybridMultiUserBeam24}
\begin{align}
&\min_{\{\bm{X}_{k}, \alpha_{k}, \beta_{k}, \psi_{k}, \phi_{k}\},\bm{Z}}~-\sum\limits_{k=1}^{K} \beta_{k}+\lambda\mbox{tr}\left(\bm{1}_{N\times N}\bm{Z}\right),\label{HybridMultiUserBeam24a}\\
&s.t.~\eqref{HybridMultiUserBeam22b}, \eqref{HybridMultiUserBeam22c}, \eqref{HybridMultiUserBeam22e},\label{HybridMultiUserBeam24b}\\
& \Phi_{k}^{\left(I\right)}\left(\psi_{k},\phi_{k}\right)\geqslant\alpha_{k}, \forall k\in\mathcal{K},\label{HybridMultiUserBeam24c}\\
&\sum\limits_{l=1,l\neq k}^{K}\overline{\gamma}_{k}\mbox{tr}\left(\bm{H}_{k}\bm{X}_{l}\right)+\overline{\gamma}_{k}\sigma^{2}
\leqslant\mbox{tr}\left(\bm{H}_{k}\bm{X}_{k}\right),\forall k\in\mathcal{K},\label{HybridMultiUserBeam24d}
\end{align}
\end{subequations}
which can be solved efficiently via a modern convex solver such as MOSEK~\cite{CLNguyen2015}. For conciseness, let $\bm{\Xi}^{\left(I\right)}$ denote the set of all variables in problem~(\ref{HybridMultiUserBeam24}) at the $I$th iteration. Algorithm~\ref{HybridMultiUserBeamA01} outlines an iterative procedure for finding a solution to problem~\eqref{HybridMultiUserBeam21} (or equivalently~\eqref{HybridMultiUserBeam22}) with a fixed $\lambda$, where $\tau$ denotes the objective of problem~\eqref{HybridMultiUserBeam21}.
\begin{algorithm}
\caption{Joint Codeword Selection and Precoder Optimization with fixed $\lambda$}\label{HybridMultiUserBeamA01}
\begin{algorithmic}[1]
\STATE Let $I=0$, generate initial points $\bm{\Xi}^{\left(I\right)}$ and compute $\tau^{(I)}$.\label{HybridMultiUserBeamA0101}
\STATE Solve (\ref{HybridMultiUserBeam24}) with $\bm{\Xi}^{\left(I\right)}$, then obtain $\bm{\Xi}^{\left(*\right)}$.\label{HybridMultiUserBeamA0102}
\STATE If $\left|\tau^{(*)}-\tau^{(I)}\right|\leq\zeta$, then output $\bm{\Xi}^{\left(*\right)}$, $\tau^{(*)}$, and stop iteration.
Otherwise, $I\leftarrow I+1$, $\tau^{(I)}\leftarrow\tau^{(*)}$, $\bm{\Xi}^{\left(I\right)}\leftarrow\bm{\Xi}^{\left(*\right)}$, and go to step~\ref{HybridMultiUserBeamA0102}.\label{HybridMultiUserBeamA0103}
\end{algorithmic}
\end{algorithm}

Problem~\eqref{HybridMultiUserBeam24} consists of a linear objective function, $K\left(M^{2}+1\right)$ positive-semidefinite constraints, $5K$ linear inequality constraints, and one convex constraint. It can be solved via convex optimization methods, such the interior point method~\cite{BoolBoyd2004}. The interior point method will take $\mathcal{O}\left(\sqrt{KM}\log\left(\epsilon\right)\right)$ iterations, where the parameter $\epsilon$ represents the solution accuracy at the algorithm's termination. In each iteration, the complexity of solving~\eqref{HybridMultiUserBeam24} is $\left(\left(M^{6}+64\right)K^3+6K^{2}M^{2}\right)$~\cite{BookYe1997}. The optimal solution returned at the  $I$th iteration is also feasible for the problem at the $\left(I + 1\right)$th iteration, as a result of the approximation in~\eqref{HybridMultiUserBeam24c}. Hence, Algorithm~\ref{HybridMultiUserBeamA01} yields a nondecreasing sequence. Since the objective of problem~(\ref{HybridMultiUserBeam21}) is bounded under the limited transmit power, the convergence of Algorithm~\ref{HybridMultiUserBeamA01} is guaranteed~\cite{GlagBibby1974}. In addition, following the similar arguments in~\cite{OperalMarks1977}, it can be proved that Algorithm~\ref{HybridMultiUserBeamA01} converges to a Karush-Kunhn-Tuker (KKT) solution of problem~(\ref{HybridMultiUserBeam22})~\cite{OperMarks1978,TSPHan2013}. To obtain a good initial point $\bm{\Xi}^{\left(0\right)}$ for Algorithm~\ref{HybridMultiUserBeamA01}, one can solve the problem~\eqref{HybridMultiUserBeam25} which was extensively studied in~\cite{TSPMehanna2013}.
\begin{equation}\label{HybridMultiUserBeam25}
\min_{\{\bm{X}_{k}\}, \bm{Z}}~\mbox{tr}\left(\bm{1}_{N\times N}\bm{Z}\right)s.t.(\ref{HybridMultiUserBeam22e}),  (\ref{HybridMultiUserBeam24d}), \bm{X}_{k}\succeq \bm{0}, \forall k\in\mathcal{K}.
\end{equation}
\begin{remark}
\rm
Let $\bm{\Xi}^{\lambda}$ denote the optimal solution to problem~(\ref{HybridMultiUserBeam19}) with fixed $\lambda$. By definition, the nonzero diagonal entries of $\bm{Z}^{\lambda}$ correspond to the selected virtual antennas (codewords). If an entry of $\bm{Z}^{\lambda}$ is zero, then the corresponding entry in all $\bm{X}_{k}^{\lambda}, \forall k$ must be zero. Let $L^{\lambda}$ be the number of nonzero diagonal entries of $\bm{Z}^{\lambda}$. Then, the effective channel of the $k$th user is an $L^{\lambda}\times 1$ vector $\overline{\bm{h}}_{k}=\widehat{\bm{F}}^{H}\bm{h}_{k}$ where the columns of $\widehat{\bm{F}}$ are the $L^{\lambda}$ selected codewords from the RF codebook $\bm{F}$. Thus, the analog precoder $\underline{\bm{F}}$ is obtained as $\underline{\bm{F}}=\widehat{\bm{F}}^{H}$.
\end{remark}

\subsection*{B. Sparse Parameter for SRmax problem}

In the previous subsection, we have introduced a turnable sparse parameter $\lambda$ to control the sparsity of the solution of the JWSPD optimization. In this subsection, we investigate how to choose a proper $\lambda$ to satisfy the RF chain constraint $\|\ddot{\bm{g}}\|_{0}\leq S$. Note that in~\eqref{HybridMultiUserBeam21}, a larger $\lambda$ makes the entries of $\bm{Z}$ (as well as $\bm{X}_{k}, \forall k\in\mathcal{K}$) more sparse, implying that less RF chains are used. On the other side, to maximize the system SR and guarantee the target rate requirement of each user, one cannot force all entries of $\bm{X}_{k}, \forall k\in\mathcal{K}$ to be zero. Thus, the sparse parameter $\lambda$ has to be properly chosen to balance maximizing the system SR and minimizing the number of the selected virtual antennas (codewords).

It is not difficult to find that the system SR increases with the number of the RF chains. Therefore, the task of find the minimum $\lambda$ such that the RF chain constraint $\|\ddot{\bm{g}}\|_{0}\leq S$ is satisfied can be accomplished by the classical one-dimension search methods, such as the bisection method~\cite{BoolBoyd2004}. For completeness, the algorithm used to find the proper sparse parameter $\lambda$ such that $\|\ddot{\bm{g}}\|_{0}\leq S$ is summarized in Algorithm~\ref{HybridMultiUserBeamA05}, where $\bm{\Lambda}^{\lambda}$ and $\widetilde{\tau}^{\lambda}$ denote respectively the set of the solution of \eqref{HybridMultiUserBeam21} and the value of $\sum_{k=1}^{K}\beta_{k}$ with $\lambda$, $\bm{\Lambda}^{T}$ and $\widetilde{\tau}^{T}$ denote respectively the set of the temporary solution of \eqref{HybridMultiUserBeam21} and the temporary value of $\sum_{k=1}^{K}\beta_{k}$ with $\lambda$. Note that the initialization of Algorithm~\ref{HybridMultiUserBeamA05} can also be finished by solving~(\ref{HybridMultiUserBeam25}).
\begin{algorithm}
\caption{SRmax Optimization for Hybrid JWSPD}\label{HybridMultiUserBeamA05}
\begin{algorithmic}[1]
\STATE Generate initial points $\bm{\Lambda}^{0}$ via solving (\ref{HybridMultiUserBeam25}), and computer $\widetilde{\tau}^{T}$. Let $flag=1$.\label{HybridMultiUserBeamS0501}
\WHILE{$flag$}
\STATE Let $\lambda=\frac{\lambda_{L}+\lambda_{U}}{2}$. 
\STATE Solve (\ref{HybridMultiUserBeam21}) with $\lambda$ and Algorithm~\ref{HybridMultiUserBeamA01}, then obtain $\bm{\Lambda}^{\lambda}$ and $\widetilde{\tau}^{\lambda}$.\label{HybridMultiUserBeamS0502}
\STATE If $L^{\lambda}>S$, let $\lambda_{L}=\lambda$, otherwise, let $\lambda_{U}=\lambda$, $\widetilde{\tau}^{\lambda}\leftarrow\sum\limits_{k=1}^{K} \beta_{k}^{\lambda}$.\label{HybridMultiUserBeamS0503}
\STATE If $\left|\widetilde{\tau}^{\lambda}-\widetilde{\tau}^{T}\right|\leq\zeta$ and $L^{\lambda}\leq S$, then let $flag=0$ and output $\bm{\Lambda}^{\lambda}$. Otherwise, $\bm{\Lambda}^{T}\leftarrow\bm{\Lambda}^{\lambda}$,
$\widetilde{\tau}^{T}\leftarrow\widetilde{\tau}^{\lambda}$.\label{HybridMultiUserBeamS0504}
\ENDWHILE
\end{algorithmic}
\end{algorithm}

\subsection*{C. Refined Solution for SRmax Problem}

Recall that in the previous subsections, the $\ell_{0}$(quasi)-norm has been approximated by the mixed $\ell_{1,\infty}$-norm squared to obtain a tractable solution. In addition, due to dropping the nonconvex rank constraint in~\eqref{HybridMultiUserBeam19}, the solution $\bm{X}_{k}, \forall k\in\mathcal{K}$ obtained by solving~\eqref{HybridMultiUserBeam21} may not be rank one. Thus, the solution provided by~\eqref{HybridMultiUserBeam21} has to be refined to fit the original problem~\eqref{HybridMultiUserBeam15}. For this purpose, after obtaining an approximate solution to~(\ref{HybridMultiUserBeam19}), we propose to solve a size-reduced SRmax problem as the last step, omitting the antennas corresponding to the zero diagonal entries of the approximated sparse solution $\bm{Z}$. The size-reduced SRmax problem is given by:
\begin{subequations}\label{HybridMultiUserBeam26}
\begin{align}
&\max_{\left\{\overline{\bm{g}}_{k}\right\}}~\sum\limits_{k=1}^{K} \overline{R}_{k},\label{HybridMultiUserBeam26a}\\
s.t.~&\overline{\mbox{SINR}}_{k}\geq \overline{\gamma}_{k}, \forall k\in\mathcal{K},\label{HybridMultiUserBeam26b}\\
&\sum\limits_{k=1}^{K}\left\|\widehat{\bm{F}}\overline{\bm{g}}_{k}\right\|_{2}^2\leq P,\label{HybridMultiUserBeam26c}
\end{align}
\end{subequations}
where $\overline{R}_{k}=\log\left(1+\overline{\mbox{SINR}}_{k}\right)$, and $\overline{\mbox{SINR}}_{k}$ is given by
\begin{equation}\label{HybridMultiUserBeam27}
\overline{\mbox{SINR}}_{k}\triangleq\frac{\left\|\overline{\bm{h}}_{k}^{H}\overline{\bm{g}}_{k}\right\|_{2}^{2}}{\sum\limits_{l=1,l\neq k}^{K}\left\|\overline{\bm{h}}_{k}^{H}\overline{\bm{g}}_{l}\right\|_{2}^{2}+\sigma^{2}}.
\end{equation}
Similarly, the size-reduced SRmax problem~\eqref{HybridMultiUserBeam26} can be equivalently reformulated as:
\begin{subequations}\label{HybridMultiUserBeam28}
\begin{align}
&\max_{\left\{\overline{\bm{g}}_{k},\overline{\alpha}_{k}, \overline{\beta}_{k}, \overline{\phi}_{k}\right\}}~\sum\limits_{k=1}^{K}\overline{\beta}_{k},\label{HybridMultiUserBeam28a}\\
s.t.&1+\overline{\alpha}_{k}\geqslant e^{\overline{\beta}_{k}}, \forall k\in\mathcal{K}, \sum\limits_{k=1}^{K}\left\|\widehat{\bm{F}}\overline{\bm{g}}_{k}\right\|_{2}^2\leq P, \label{HybridMultiUserBeam28b}\\
&\frac{{\left\|\overline{\bm{h}}_{k}^{H}\overline{\bm{g}}_{k}\right\|_{2}^{2}}}{\overline{\phi}_{k}}\geq \overline{\gamma}_{k}, \forall \frac{{\left\|\overline{\bm{h}}_{k}^{H}\overline{\bm{g}}_{k}\right\|_{2}^{2}}}{\overline{\phi}_{k}}\geq \overline{\alpha}_{k}, \forall k\in\mathcal{K}, \label{HybridMultiUserBeam28c}\\ %
&\sum\limits_{l=1,l\neq k}^{K}\left\|\overline{\bm{h}}_{k}^{H}\overline{\bm{g}}_{l}\right\|_{2}^{2}
+\sigma^{2}\leqslant\overline{\phi}_{k}, \forall k\in\mathcal{K}.\label{HybridMultiUserBeam28d}
\end{align}
\end{subequations}
Similar to the problem~\eqref{HybridMultiUserBeam22}, \eqref{HybridMultiUserBeam28} is also a nonconvex problem due to the constraints in~\eqref{HybridMultiUserBeam28c}. For \eqref{HybridMultiUserBeam28c}, we have the following the convex low boundary:
\begin{equation}\label{HybridMultiUserBeam29}
\begin{split}
&\frac{\left\|\overline{\bm{h}}_{k}^{H}\overline{\bm{g}}_{k}\right\|_{2}^{2}}
{\overline{\phi}_{k}}\geq\overline{\Phi}_{k}^{\left(I\right)}\left(\overline{\bm{g}}_{k},\overline{\phi}_{k}\right)
\triangleq\\
& \frac{2\Re\left(\left(\overline{\bm{g}}_{k}^{\left(I\right)}\right)^{H}\overline{\bm{h}}_{k}\overline{\bm{h}}_{k}^{H}
\overline{\bm{g}}_{k}\right)}{\overline{\phi}_{k}^{\left(I\right)}}
-\left(\frac{\left\|\overline{\bm{h}}_{k}^{H}\overline{\bm{g}}_{k}^{\left(I\right)}\right\|_{2}}
{\overline{\phi}_{k}^{\left(I\right)}}\right)^{2}\overline{\phi}_{k}, \forall k\in\mathcal{K},
\end{split}
\end{equation}
where $I$ denotes the $I$th iteration. Thus, the constraints in~\eqref{HybridMultiUserBeam28c} can be approximated as:
\begin{equation}\label{HybridMultiUserBeam30}
\overline{\Phi}_{k}^{\left(I\right)}\left(\overline{\bm{g}}_{k},\overline{\phi}_{k}\right)\geq \overline{\gamma}_{k},  \overline{\Phi}_{k}^{\left(I\right)}\left(\overline{\bm{g}}_{k},\overline{\phi}_{k}\right)\geq \overline{\alpha}_{k}, \forall k\in\mathcal{K}.
\end{equation}
Consequently we can obtain a stationary solution to~(\ref{HybridMultiUserBeam28}), by solving the following series of convex problems:
\begin{equation}\label{HybridMultiUserBeam31}
\max_{\left\{\overline{\bm{g}}_{k},\overline{\alpha}_{k}, \overline{\beta}_{k}, \overline{\phi}_{k}\right\}}~\sum\limits_{k=1}^{K}\overline{\beta}_{k},~
s.t.~\eqref{HybridMultiUserBeam28b}, \eqref{HybridMultiUserBeam28d}, \eqref{HybridMultiUserBeam30}.
\end{equation}
Such an iterative procedure is outlined in Algorithm~\ref{HybridMultiUserBeamA03}, where $\overline{\bm{\Xi}}^{\left(I\right)}$ and $\overline{\tau}^{\left(I\right)}$ denote the set of the solution and the objective value of problem~(\ref{HybridMultiUserBeam31}) at the $I$th iteration, respectively. The convergence property of Algorithm~\ref{HybridMultiUserBeamA03} is similar with that of Algorithm~\ref{HybridMultiUserBeamA01}. The computational complexity of Algorithm~\ref{HybridMultiUserBeamA03} is about $\mathcal{O}\left(M^{4}K^{4}\right)$~\cite{CLNguyen2015}.
\begin{algorithm}
\caption{Transmit Beamforming Optimization}\label{HybridMultiUserBeamA03}
\begin{algorithmic}[1]
\STATE Let $I=0$, generate initial points $\overline{\bm{\Xi}}^{\left(I\right)}$ and compute $\overline{\tau}^{(I)}$.\label{HybridMultiUserBeamA0301}
\STATE Solve (\ref{HybridMultiUserBeam31}) with $\overline{\bm{\Xi}}^{\left(I\right)}$, then obtain $\overline{\bm{\Xi}}^{\left(*\right)}$.\label{HybridMultiUserBeamA0302}
\STATE If $\left|\overline{\tau}^{(*)}-\overline{\tau}^{(I)}\right|\leq\zeta$, then output $\overline{\bm{\Xi}}^{\left(*\right)}$ and stop iteration. Otherwise, $I\leftarrow I+1$, $\overline{\tau}^{(I)}\leftarrow\overline{\tau}^{(*)}$, $\overline{\bm{\Xi}}^{\left(I\right)}\leftarrow\overline{\bm{\Xi}}^{\left(*\right)}$, and go to step~\ref{HybridMultiUserBeamA0302}.\label{HybridMultiUserBeamA0303}
\end{algorithmic}
\end{algorithm}

In the next,  we investigate how to obtain a good initial point for Algorithm~\ref{HybridMultiUserBeamA03}. Let $\overline{\bm{g}}_{k}=\sqrt{q_{k}}\widehat{\bm{g}}_{k}$, $\forall k\in\mathcal{K}$\footnote{It is easy to find that (\ref{HybridMultiUserBeam34}) is a weighted sum power minimization problem which can be regarded as an extension of  the conventional power minimization problem.}. We propose to use the solution of the following problem as the initial point:
\begin{equation}\label{HybridMultiUserBeam32}
\min_{\left\{q_{k}, \widehat{\bm{g}}_{k}\right\}}\sum\limits_{k=1}^{K}q_{k}\widehat{\bm{g}}_{k}^{H}\widehat{\bm{F}}^{H}\widehat{\bm{F}}\widehat{\bm{g}}_{k}~s.t.~\overline{\mbox{SINR}}_{k}\geq \overline{\gamma}_{k}, \left\|\widehat{\bm{g}}_{k}\right\|_{2}^{2}=1, \forall k\in\mathcal{K}.
\end{equation}
We can show that problem~(\ref{HybridMultiUserBeam32}) is dual to the following virtual uplink problem~\cite{TVTSchubert2004,TCOMHe2015}:
\begin{equation}\label{HybridMultiUserBeam33}
\min_{\left\{p_{k}, \widehat{\bm{g}}_{k}\right\}}\sigma^{2}\sum\limits_{k=1}^{K}p_{k}~s.t.~\overleftarrow{\mbox{SINR}}_{k}\geq \overline{\gamma}_{k}, \left\|\widehat{\bm{g}}_{k}\right\|_{2}^{2}=1, \forall k\in\mathcal{K},
\end{equation}
where $\widehat{\bm{g}}_{k}$ can be regarded as the combiner of the dual uplink channel, $p_{k}$ has the interpretation of being the dual uplink power $k$th user in the virtual uplink, and $\overleftarrow{\mbox{SINR}}_{k}$ is given by
\begin{equation}\label{HybridMultiUserBeam34}
\overleftarrow{\mbox{SINR}}_{k}\triangleq\frac{p_{k}\left\|\overline{\bm{h}}_{k}^{H}\widehat{\bm{g}}_{k}\right\|_{2}^{2}}{\sum\limits_{l=1,l\neq k}^{K}p_{l}\left\|\overline{\bm{h}}_{l}^{H}\widehat{\bm{g}}_{k}\right\|_{2}^{2}+\widehat{\bm{g}}_{k}^{H}\widehat{\bm{F}}^{H}\widehat{\bm{F}}\widehat{\bm{g}}_{k}}.
\end{equation}
Furthermore, when the optimal solutions of problems~(\ref{HybridMultiUserBeam32}) and (\ref{HybridMultiUserBeam33}) are obtained, we have $\sum\limits_{k=1}^{K}q_{k}\widehat{\bm{g}}_{k}^{H}\widehat{\bm{F}}^{H}\widehat{\bm{F}}\widehat{\bm{g}}_{k}=\sigma^{2}\sum\limits_{k=1}^{K}p_{k}$. It was shown in~\cite{TITZhang2012,TCOMHe2015,TWCSadek2007} that the solution $\left\{\widehat{\bm{g}}_{k}\right\}$ of (\ref{HybridMultiUserBeam33}) is given by
\begin{equation}\label{HybridMultiUserBeam35}
\widehat{\bm{g}}_{k}^{*}\propto \mbox{max. eigenvector} \left(\left(\sum\limits_{l=1,l\neq k}^{K}p_{l}\overline{\bm{H}}_{l}+\widehat{\bm{F}}^{H}\widehat{\bm{F}}\right)^{-1}\overline{\bm{H}}_{k}\right).
\end{equation}
Thus, the algorithm used to solve (\ref{HybridMultiUserBeam33}) is summarized in Algorithm~\ref{HybridMultiUserBeamA04} with provable convergence~\cite{TSPWiessel2006}.
\begin{algorithm}
\caption{Transmit Beamforming Initialization}\label{HybridMultiUserBeamA04}
\begin{algorithmic}[1]
\STATE Initialize beamforming vector $\left\{\widehat{\bm{g}}_{k}\right\}$.\label{HybridMultiUserBeamA0401}
\STATE Optimize $\left\{p_{k}\right\}$ by first finding the fixed-point $p_{k}^{*}$ of the following equation by iterative function evaluation:\label{HybridMultiUserBeamA0402}
\begin{equation}
p_{k}^{*}=\overline{\gamma}_{k}\frac{\sum\limits_{l=1,l\neq k}^{K}p_{l}\left\|\overline{\bm{h}}_{l}^{H}\widehat{\bm{g}}_{k}\right\|_{2}^{2}
+\widehat{\bm{g}}_{k}^{H}\widehat{\bm{F}}^{H}\widehat{\bm{F}}\widehat{\bm{g}}_{k}}{\left\|\overline{\bm{h}}_{k}^{H}\widehat{\bm{g}}_{k}
\right\|_{2}^{2}} \nonumber
\end{equation}
\STATE Find the optimal uplink beamformers based on the optimal uplink power allocation $p_{k}^{*}$ with (\ref{HybridMultiUserBeam36}).\label{HybridMultiUserBeamA0403}
\STATE Repeat steps \ref{HybridMultiUserBeamA0402} and \ref{HybridMultiUserBeamA0403} until convergence.
\end{algorithmic}
\end{algorithm}

To find $\left\{q_{k}\right\}$ in terms of $\left\{\widehat{\bm{g}}_{k}\right\}$ that is obtained from the virtual uplink channel, i.e., (\ref{HybridMultiUserBeam35}), we note that the SINR constraints in (\ref{HybridMultiUserBeam32}) must be all actived at the global optimum point. So
\begin{equation}\label{HybridMultiUserBeam36}
q_{k}=\sum\limits_{l=1,l\neq k}^{K}q_{l}\frac{\overline{\gamma}_{k}}{\left\|\overline{\bm{h}}_{k}^{H}\widehat{\bm{g}}_{k}\right\|_{2}^{2}}
\left\|\overline{\bm{h}}_{k}^{H}\widehat{\bm{g}}_{l}\right\|_{2}^{2}+\sigma^{2}
\frac{\overline{\gamma}_{k}}{\left\|\overline{\bm{h}}_{k}^{H}\widehat{\bm{g}}_{k}\right\|_{2}^{2}}, \forall k\in\mathcal{K}.
\end{equation}
Thus, we obtain a set of $K$ linear equations with $K$ unknowns $\left\{q_{k}\right\}$, which can be solved as
\begin{equation}\label{HybridMultiUserBeam37}
\bm{q}=\bm{\Psi}\bm{G}\bm{q}+\sigma^{2}\bm{\Psi}\mathds{1}_{K},
\end{equation}
where $\bm{q}=\left[q_{1},\cdots,q_{K}\right]^{T}$, $\bm{\Psi}=diag\left\{\frac{\overline{\gamma}_{1}}{\left\|\overline{\bm{h}}_{1}^{H}\widehat{\bm{g}}_{1}\right\|_{2}^{2}},\cdots,
\frac{\overline{\gamma}_{K}}{\left\|\overline{\bm{h}}_{K}^{H}\widehat{\bm{g}}_{K}\right\|_{2}^{2}}\right\}$, $\bm{G}\left(k,k\right)=0$ and $\bm{G}\left(k,l\right)=\left\|\overline{\bm{h}}_{k}^{H}\widehat{\bm{g}}_{l}\right\|_{2}^{2}$ for $k\neq l$. Defining an extended power vector $\widetilde{\bm{q}}=\left[\bm{q}^{T},1\right]^{T}$ and an extended coupling matrix
\begin{equation}\label{HybridMultiUserBeam38}
\bm{Q}=\begin{bmatrix}
\bm{\Psi}\bm{G}&\bm{\Psi}\mathds{1}_{K}\\
\frac{1}{P_{max}}\bm{a}^{T}\bm{\Psi}\bm{G}& \frac{1}{P_{max}}\bm{a}^{T}\bm{\Psi}\mathds{1}_{K}
\end{bmatrix}.
\end{equation}
where $P_{max}=\sigma^{2}\sum\limits_{k=1}^{K}p_{k}$, $\bm{a}^{T}=\left[a_{1},\cdots,a_{K}\right]$, $a_{k}=\widehat{\bm{g}}_{k}^{H}\widehat{\bm{F}}^{H}\widehat{\bm{F}}\widehat{\bm{g}}_{k}, \forall k$.
According to the conclusions in~\cite{TVTSchubert2004}, we can easily obtain the optimal power vector $\bm{q}$ as the first $K$ components of the dominant eigenvector of $\bm{Q}$, which can be scaled such that its last component equals one. The solution for $\left\{q_{k}\right\}$, combined with that for $\left\{\widehat{\bm{g}}_{k}\right\}$, gives an explicit solution of the beamforming vector $\left\{\overline{\bm{g}}_{k}\right\}$ via an virtual uplink channel. Once the beamforming vector $\left\{\overline{\bm{g}}_{k}\right\}$ is obtained, the baseband beamforming vector $\underline{\bm{g}}_{k}$ is obtained, as $\underline{\bm{g}}_{k}=\left[\left\{\overline{\bm{g}}_{k}\right\}^{T},\bm{0}_{\left(S-L^{\lambda}\right),1}^{T}\right]^{T}$. In fact, the remaining $S-L^{\lambda}$ RF chains with the corresponding phase shifter networks can be turned off to improve the system EE.

\section*{\sc \uppercase\expandafter{\romannumeral5}. Joint Codeword Selection and Precoder Optimization for EEmax Problem}

In this section, we consider the EEmax problem~\eqref{HybridMultiUserBeam14}, which is more difficult than the SRmax problem. Indeed the objective in~\eqref{HybridMultiUserBeam14} is given by a more complex fractional form, and the $\ell_{0}$-(quasi)norm appears not only in the constraint but also in the denominator of the objective. To find the globally optimal solution to~(\ref{HybridMultiUserBeam14}) requires an exhaustive search over all  $\sum\limits_{l=L_{Min}}^{S}\left(N\atop l\right)$ possible sparse patterns of $\ddot{\bm{g}}$, where $L_{Min}\leqslant S$ is the minimum number of the selected RF chains that can achieve the target rate requirement of each user under the power constraint. Unfortunately, for each pattern of $\ddot{\bm{g}}$, (\ref{HybridMultiUserBeam14}) is an NP-hard problem. Thus, we seek a practical and efficient method to address the EEmax problem~\eqref{HybridMultiUserBeam14}.

\subsection*{A. Joint Codeword Selection and Precoder Design for EEmax problem}

Similarly, we first use the convex squared $\ell_{1,\infty}$-norm to approximate the nonconvex $\ell_{0}$-(quasi)norm in the power consumption term $P_{dyn}$. Then, we also introduce a turnable sparse parameter $\lambda\geq 0$ as a group-sparsity inducing regularization to control the sparsity of the solution so that the RF chain constraint~\eqref{HybridMultiUserBeam14c} can be temporarily omitted for fixed $\lambda$. By doing so, problem~\eqref{HybridMultiUserBeam14} can be relaxed as:

\begin{subequations}\label{HybridMultiUserBeam39}
\begin{align}
&\max_{\{\bm{X}_{i,j}\}}~\frac{\sum\limits_{k=1}^{K} R_{k}}
{\epsilon\sum\limits_{k=1}^{K}\mbox{tr}\left(\widetilde{\bm{F}}\bm{X}_{k,k}\right)+P_{dyn}\left(\lambda\right)},\label{HybridMultiUserBeam39a}\\
s.t.~&\mbox{SINR}_{k}\geq \overline{\gamma}_{k}, \forall k,
~\sum\limits_{k=1}^{K}\mbox{tr}\left(\widetilde{\bm{F}}\bm{X}_{k,k}\right)\leq P,\label{HybridMultiUserBeam39b}\\
&\bm{X}_{k,k}\succeq \bm{0}, \forall k\in\mathcal{K},
\end{align}
\end{subequations}
where the nonconvex $\mbox{rank}\left(\bm{X}_{i,j}\right)=1$, $\forall i,j$ constraints are dropped, and the dynamic power consumption is given by
\begin{equation}\label{HybridMultiUserBeam40}
P_{dyn}\left(\lambda\right)=f\left(\lambda\right)\sum\limits_{n=1}^{N}\sum\limits_{m=1}^{N}
\max_{i,j\in\left\{1,\cdots,K\right\}}\left|\bm{X}_{i,j}\left(n,m\right)\right|+P_{sta},
\end{equation}
where $f\left(\lambda\right)=P_{RFC}+M P_{PS}+P_{DAC}+\lambda$. Note that $\bm{X}_{i,j}$, $\forall i\neq j$, only appear in the power consumption item $P_{dyn}\left(\lambda\right)$. Therefore, similar to Theorem~\ref{HybridMultiUserBeamL01}, we have the following result.
\begin{theorem}\label{HybridMultiUserBeamL02}
Let $\breve{\bm{X}}_{i,j}, \forall i,j\in\mathcal{K}$ be the optimal solution of~(\ref{HybridMultiUserBeam39}), then $\breve{\bm{X}}_{i,i}, \forall i\in\mathcal{K}$ with $\bm{X}_{i,j}=\bm{0}$, $\forall i\neq j, i,j\in\mathcal{K}$ is also the optimal solution of~(\ref{HybridMultiUserBeam39}).
\end{theorem}
\begin{proof}
First, we prove that the inequalities $\left|\breve{\bm{X}}_{i,j}\right|\leqslant\left|\breve{\bm{X}}_{i,i}\right|, \forall i\neq j$ hold. Suppose that there is one pair of indices $\left(i_{0},j_{0}\right), i_{0}\neq j_{0}$ and $\left(n_{0},m_{0}\right)$ such that $\left|\breve{\bm{X}}_{i,j}\left(n,m\right)\right|\leqslant\left|\breve{\bm{X}}_{i,i}\left(n,m\right)\right|, \forall i\neq j,  n, m$ except for $\left|\breve{\bm{X}}_{i_{0},j_{0}}\left(n_{0},m_{0}\right)\right|\geq
\left|\breve{\bm{X}}_{i,i}\left(n_{0},m_{0}\right)\right|, \forall k$. Let $\overline{\bm{X}}_{i,j}, \forall i,j\in\mathcal{K}$ be another solution obtained by letting $\overline{\bm{X}}_{i,j}\left(n,m\right)=\breve{\bm{X}}_{i,j}\left(n,m\right), \forall i, j, n, m$ except for $\overline{\bm{X}}_{i_{0},j_{0}}\left(n_{0},m_{0}\right)=0$. Note that $\bm{X}_{i,j},\forall i\neq j$ only appear in the constraints~\eqref{HybridMultiUserBeam39b}. Thus, $\overline{\bm{X}}_{i,j}, \forall i,j\in\mathcal{K}$ is a feasible solution to problem~(\ref{HybridMultiUserBeam39}) and satisfies the following inequality~\eqref{HybridMultiUserBeam20}.
\begin{equation}\label{HybridMultiUserBeam20}
\begin{split}
\breve{P}_{dyn}\left(\lambda\right)&=f\left(\lambda\right)\left(\left|\breve{\bm{X}}_{i_{0},j_{0}}\left(n_{0},m_{0}\right)\right|+
\underbrace{\sum\limits_{n=1}^{N}\sum\limits_{m=1}^{N}}_{ \left(n,m\right)\neq\left(n_{0},m_{0}\right)}\max_{i}\left|\breve{\bm{X}}_{i,i}\left(n,m\right)\right|\right)+P_{sta}\\
&>\overline{P}_{dyn}\left(\lambda\right)=f\left(\lambda\right)\sum\limits_{n=1}^{N}\sum\limits_{m=1}^{N}
\max_{i}\left|\overline{\bm{X}}_{i,i}\left(n,m\right)\right|+P_{sta}.
\end{split}
\end{equation}
Note that $\breve{\bm{X}}_{i,j}$ and $\overline{\bm{X}}_{i,j}$, $\forall i,j\in\mathcal{K}$ achieve the same user rate. Combining the objective of problem~\eqref{HybridMultiUserBeam39} and~\eqref{HybridMultiUserBeam20}, we can obtain a better objective by using $\overline{\bm{X}}_{i,j}, \forall i,j\in\mathcal{K}$ than using $\breve{\bm{X}}_{i,j}, \forall i,j\in\mathcal{K}$, which is a contradiction. Therefore, we have $\left|\breve{\bm{X}}_{i,j}\right|\leqslant\left|\breve{\bm{X}}_{i,i}\right|, \forall i\neq j$.

Note that $\bm{X}_{i,j}$, $\forall i\neq j$, only appear in the power consumption item $P_{dyn}\left(\lambda\right)$. Combining $\left|\breve{\bm{X}}_{i,j}\right|\leqslant\left|\breve{\bm{X}}_{i,i}\right|, \forall i\neq j$ with~\eqref{HybridMultiUserBeam40}, one can easily see that the power consumption item $P_{dyn}\left(\lambda\right)$ dose not change by setting $\bm{X}_{i,j}=\bm{0}$, $\forall i\neq j$. Consequently, $\breve{\bm{X}}_{i,i}, \forall i\in\mathcal{K}$ with $\bm{X}_{i,j}=\bm{0}$, $\forall i\neq j$ are still optimal.
\end{proof}

Theorem~\ref{HybridMultiUserBeamL02} also indicates that we can simplify problem~\eqref{HybridMultiUserBeam39} by setting $\bm{X}_{i,j}=\bm{0}$, $\forall i\neq j$ without any loss of optimality. Hence, similar to the transformation between~\eqref{HybridMultiUserBeam19} and~\eqref{HybridMultiUserBeam21}, (\ref{HybridMultiUserBeam39}) is equivalent to
\begin{subequations}\label{HybridMultiUserBeam41}
\begin{align}
&\max_{\{\bm{X}_{k}\}, \bm{Z}}~\frac{\sum\limits_{k=1}^{K} R_{k}}
{\epsilon\sum\limits_{k=1}^{K}\mbox{tr}\left(\widetilde{\bm{F}}\bm{X}_{k}\right)+P_{dyn}\left(\bm{Z},\lambda\right)},\label{HybridMultiUserBeam41a}\\
s.t.~&\mbox{SINR}_{k}\geq \overline{\gamma}_{k}, \forall k\in\mathcal{K},
~\sum\limits_{k=1}^{K}\mbox{tr}\left(\widetilde{\bm{F}}\bm{X}_{k}\right)\leq P,\label{HybridMultiUserBeam41b}\\
&\bm{X}_{k}\succeq \bm{0}, \bm{Z}\geq\left|\bm{X}_{k}\right|, \forall k\in\mathcal{K},
\end{align}
\end{subequations}
where $P_{dyn}\left(\bm{Z},{\lambda}\right)=f\left(\lambda\right)\mbox{tr}\left(\bm{1}_{N\times N}\bm{Z}\right)+P_{sta}$. Introducing auxiliary variables $\alpha_{k}, \beta_{k}, \psi_{k}, \phi_{k}$, $\forall k\in\mathcal{K}$, $\tau$, $\kappa$, and $\chi$, (\ref{HybridMultiUserBeam41}) can be equivalently rewritten as
\begin{subequations}\label{HybridMultiUserBeam42}
\begin{align}
&\max_{\{\bm{X}_{k}, \alpha_{k}, \beta_{k}, \psi_{k}, \phi_{k}\},\bm{Z}, \tau, \kappa, \chi}~\chi,\label{HybridMultiUserBeam42a}\\
&s.t.~\frac{\tau^{2}}{\kappa}\geqslant\chi,  \frac{\psi_{k}^{2}}{\phi_{k}}\geqslant\alpha_{k}, \forall k\in\mathcal{K}\label{HybridMultiUserBeam42b}\\
&\sum\limits_{k=1}^{K} \beta_{k}\geqslant\tau^2, \eqref{HybridMultiUserBeam22b}, \eqref{HybridMultiUserBeam22c}, \eqref{HybridMultiUserBeam22e}\label{HybridMultiUserBeam42c}\\
&\epsilon\sum\limits_{k=1}^{K}\mbox{tr}\left(\widetilde{\bm{F}}\bm{X}_{k}\right)
+P_{dyn}\left(\bm{Z},{\lambda}\right)\leqslant\kappa\label{HybridMultiUserBeam42d}\\
&\sum\limits_{l=1,l\neq k}^{K}\overline{\gamma}_{k}\mbox{tr}\left(\bm{H}_{k}\bm{X}_{l}\right)+\overline{\gamma}_{k}\sigma^{2}
\leqslant\mbox{tr}\left(\bm{H}_{k}\bm{X}_{k}\right),\forall k\in\mathcal{K}\label{HybridMultiUserBeam42e}
\end{align}
\end{subequations}
Similarly, the difficulty of solving (\ref{HybridMultiUserBeam42}) lies in (\ref{HybridMultiUserBeam42b}), as the two constraints in~(\ref{HybridMultiUserBeam42b}) are nonconvex. Thus, we exploit the SCA method~\cite{OperalMarks1977} to approximate the two inequalities in~\eqref{HybridMultiUserBeam42b} by two convex constraints. By replacing \eqref{HybridMultiUserBeam42b} with the convex lower bounds at the $I$th  iteration, problem~\eqref{HybridMultiUserBeam42} can be approximated by the following convex program:
\begin{subequations}\label{HybridMultiUserBeam43}
\begin{align}
&\max_{\{\bm{X}_{k}, \alpha_{k}, \beta_{k}, \psi_{k}, \phi_{k}, \bm{\mu}_{k}\},\bm{Z}, \tau, \kappa, \chi}~\chi,\label{HybridMultiUserBeam43a}\\
s.t.~&(\ref{HybridMultiUserBeam42c}), (\ref{HybridMultiUserBeam42d}), \eqref{HybridMultiUserBeam42e}, \label{HybridMultiUserBeam43b}\\
&\Psi^{\left(I\right)}\left(\tau,\kappa\right)\geqslant\chi,  \Phi_{k}^{\left(I\right)}\left(\psi_{k},\phi_{k}\right)\geqslant\alpha_{k}, \forall k\in\mathcal{K},\label{HybridMultiUserBeam43c}
\end{align}
\end{subequations}
where $\Psi^{\left(I\right)}\left(\tau,\kappa\right)\triangleq
2\frac{\tau^{\left(I\right)}}{\kappa^{\left(I\right)}}\tau
-\left(\frac{\tau^{\left(I\right)}}{\kappa^{\left(I\right)}}\right)^{2}\kappa$. Thus, problem~\eqref{HybridMultiUserBeam43} can be solved via the similar procedure as described in Algorithm~\ref{HybridMultiUserBeamA01}.

\subsection*{B. Sparse Parameter for EEmax problem}

Similarly, a larger $\lambda$ leads to a more sparse solution to the (approximated) EEmax problem~\eqref{HybridMultiUserBeam39}, which corresponds to less RF chains used. On the other side, $\lambda$ cannot be infinite, which would lead to a zero solution and contradict the task of maximizing the system EE. Hence, $\lambda$ has to be properly chosen. However, unlike to the SRmax problem~\eqref{HybridMultiUserBeam21} or the total power minimization problem with RF chain constraints~\cite{TSPMehanna2013}, the system EE is not monotonic with respect to the number of RF chains or the sparse parameter $\lambda$. Indeed, the system EE is a piecewise function with respect to the sparse parameter $\lambda$, as illustrated in Fig.~\ref{Chi&LambdaComparison} and Table~\ref{Chi&LambdaComparisonList}. Consequently, the bisection method cannot be used to optimize $\lambda$~\cite{BoolBoyd2004}.

To address the above issue, we devise a dynamic interval compression method to search a suitable $\lambda$. Specifically, let $\mathcal{A}_{Min}$ be the set of the indices of the $L_{Min}$ selected virtual antennas (codewords). Let $L_{Max}$ be the number of the virtual antennas (codewords) achieving the maximum EE by ignoring the available RF chain constraint, which correspondes to $\lambda=0$, and $\mathcal{A}_{Max}$ be the set of the indices of the $L_{Max}$ selected virtual antennas (codewords) in this case. Considering that the allowable number of RF chains is a discrete value but the sparse parameter $\lambda$ is continuous, we introduce the following definition.
\begin{definition}\label{HybridMultiUserBeamD01}
For any small positive number $\epsilon$ and $\forall L\in\left\{L_{Min},\cdots, L_{Max}-1\right\}$, $\lambda_{L}^{key}$ is called a breaking point if the optimal solutions $\bm{Z}^{\lambda_{L}^{key}-\varepsilon}$ and $\bm{Z}^{\lambda_{L}^{key}}$ of~(\ref{HybridMultiUserBeam41}) have $L+1$ and $L$ nonzero diagonal entries, respectively.
\end{definition}
\begin{theorem}\label{HybridMultiUserBeamL03}
Let $\bm{\Xi}^{\lambda}$ be the solution of~(\ref{HybridMultiUserBeam41})  with $\forall \lambda\in\left[\lambda_{L+1}^{key},\lambda_{L}^{key}\right)$ and $L$, $L+1\in\left\{L_{Min}, L_{Min}+1,\cdots, L_{Max}\right\}$. Then, $\bm{Z}^{\lambda}$  has also $L+1$ nonzero diagonal entries and the inequality $\chi^{\lambda}\leqslant\chi^{\lambda_{L+1}^{key}}$ holds.
\end{theorem}
\begin{proof}
Following the definition of the breaking point $\lambda_{L+1}^{key}$, it is easy to see that $\bm{Z}^{\lambda}$ has also $L+1$ nonzero diagonal entries. If $\chi^{\lambda}>\chi^{\lambda_{L+1}^{key}}$, recalling $\lambda_{L+1}^{key}\leqslant\lambda$, then we have~\eqref{HybridMultiUserBeam44},
\begin{equation}\label{HybridMultiUserBeam44}
\begin{split}
&\chi^{\lambda_{L+1}^{key}}=\frac{\left(\tau^{\lambda_{L+1}^{key}}\right)^{2}}{\kappa^{\lambda_{L+1}^{key}}}
=\frac{\left(\tau^{\lambda_{L+1}^{key}}\right)^{2}}{\epsilon\sum\limits_{k=1}^{K}\mbox{tr}
\left(\widetilde{\bm{F}}\bm{X}_{k}^{\lambda_{L+1}^{key}}\right)+f\left(\lambda_{L+1}^{key}\right)\mbox{tr}\left(\bm{1}_{N\times N}\bm{Z}^{\lambda_{L+1}^{key}}\right)+P_{sta}}\\
&<\chi^{\lambda}=\frac{\left(\tau^{\lambda}\right)^{2}}{\kappa^{\lambda}}
=\frac{\left(\tau^{\lambda}\right)^{2}}{\epsilon\sum\limits_{k=1}^{K}\mbox{tr}
\left(\widetilde{\bm{F}}\bm{X}_{k}^{\lambda}\right)+f\left(\lambda\right)\mbox{tr}\left(\bm{1}_{N\times N}\bm{Z}^{\lambda}\right)+P_{sta}}\\
&<\frac{\left(\tau^{\lambda}\right)^{2}}{\epsilon\sum\limits_{k=1}^{K}\mbox{tr}
\left(\widetilde{\bm{F}}\bm{X}_{k}^{\lambda}\right)+f\left(\lambda_{L+1}^{key}\right)\mbox{tr}\left(\bm{1}_{N\times N}\bm{Z}^{\lambda}\right)+P_{sts}},
\end{split}
\end{equation}
which contradicts the fact that $\bm{\Xi}^{\lambda_{L+1}^{key}}$ is the optimal solution to problem~(\ref{HybridMultiUserBeam43}) with fixed $\lambda_{L+1}^{key}$. Thus, the conclusions given in Theorem~\ref{HybridMultiUserBeamL03} are proven.
\end{proof}

According the definition of the breaking point and the non-monotonic property of the system EE with respect to $\lambda$, one shall find the values of all breaking points. Let $\bm{Z}^{\lambda_{i}}$ be the solution of problem~\eqref{HybridMultiUserBeam41} with fixed $\lambda_{i}$, $i=1,2$. Let $\bm{Z}^{\lambda}$ be the solution of problem~\eqref{HybridMultiUserBeam41} for $\forall \lambda\in\left[\lambda_{1},\lambda_{2}\right]$. Theorem~\ref{HybridMultiUserBeamL03} implies that if $\bm{Z}^{\lambda_{1}}$ and $\bm{Z}^{\lambda_{2}}$ have the same number of the nonzero diagonal entries, $\bm{Z}^{\lambda_{1}}$, $\bm{Z}^{\lambda_{2}}$, and $\bm{Z}^{\lambda}$ have the same number of the nonzero diagonal entries. Based on this result, we propose a one-dimension dynamic interval compression method, which is summarized in Algorithm~\ref{TurningParaOptimization}, to find a suitable sparse parameter $\lambda$ and obtain the corresponding codewords. Note that in Algorithm~\ref{TurningParaOptimization}, $\mathcal{A}^{\lambda}$ denotes the set of the indices of the selected virtual antennas (codewords) with fixed $\lambda$ and $\varrho^{\lambda}$ is calculated as
$$\frac{\sum\limits_{k=1}^{K}\log\left(1+\frac{tr\left(\bm{H}_{k}\bm{X}_{k}\right)}
{\sum\limits_{l=1, l\neq k}^{K}tr\left(\bm{H}_{k}\bm{X}_{l}\right)+\sigma_{k}^{2}}\right)}
{\xi\sum\limits_{k=1}^{K}tr\left(\bm{X}_{k}\right)
+L^{\lambda}\left(P_{RCF}+M P_{PS}+P_{DAC}\right)+P_{sta}}.
$$
The initialization of Algorithm~\ref{TurningParaOptimization} can also be obtained by solving~(\ref{HybridMultiUserBeam25}) and letting other constraints to be activated. In addition, $\lambda_{U}$ should be large enough such that the number of the active RF chains equals to $L_{Min}$.
According to Theorem~\ref{HybridMultiUserBeamL03}, if two intervals in the intervals set $\mathcal{I}$ have an intersection in the intervals set $\mathcal{I}$, they shall be combined into one interval, for example $\left[100,150\right]$ and $\left[150,200\right]$ are combined to $\left[100,200\right]$.

\begin{algorithm}
\caption{RF Chains Set Generation: Part \uppercase\expandafter{\romannumeral1}}\label{TurningParaOptimization}
\begin{algorithmic}[1]
\STATE Initialize $\bm{\Xi}^{Opt}$, $\mathcal{L}_{\lambda}=\emptyset$, $\mathcal{A}=\emptyset$, $L_{c}=0$, $L_{0}=S$, $L_{1}=L_{Min}$, $Flag=1$.
\STATE Let $\lambda_{L}=0$, solve (\ref{HybridMultiUserBeam15}) with Algorithm~\ref{HybridMultiUserBeamA01} and $\lambda_{L}$, obtain $\bm{\Xi}^{\lambda_{L}}$, $L^{\lambda_{L}}$, and $\varrho^{(L)}$.
\STATE Let $\lambda_{U}$ be a larger positive number, solve (\ref{HybridMultiUserBeam15}) with Algorithm~\ref{HybridMultiUserBeamA01} and $\lambda_{U}$, obtain $\bm{\Xi}^{\lambda_{U}}$, $L^{\lambda_{U}}$, and $\varrho^{(U)}$.
\STATE Let $\lambda_{Temp}=0$, $L_{Temp}=L^{\lambda_{L}}$, and $\mathcal{I}=\emptyset$ be a set of intervals.
\IF{$L^{\lambda_{L}}\leqslant S$}
\STATE $\mathcal{L}=\mathcal{L}\cup\left\{\left(L^{\lambda_{L}},\lambda_{L}\right)\right\}$, $L_{c}=L_{c}+1$, $L_{0}=L^{\lambda_{L}}$.
\STATE Let $\mathcal{A}_{L^{\lambda_{L}}}$ be the set of the selected codewords indices set obtained by $\bm{\Xi}^{\lambda_{L}}$, $\mathcal{A}
=\mathcal{A}\cup\left\{\mathcal{A}_{L^{\lambda_{L}}}\right\}$.
\ENDIF
\IF{$L^{\lambda_{U}}\leqslant S$}
\STATE $\mathcal{L}=\mathcal{L}\cup\left\{\left(L^{\lambda_{U}},\lambda_{U}\right)\right\}$, $L_{c}=L_{c}+1$, $L_{1}=L^{\lambda_{U}}$.
\STATE Let $\mathcal{A}_{L^{\lambda_{U}}}$ be the set of the selected codewords indices set obtained by $\bm{\Xi}^{\lambda_{U}}$, $\mathcal{A}
=\mathcal{A}\cup\left\{\mathcal{A}_{L^{\lambda_{U}}}\right\}$.
\ENDIF
\IF{$L_{0}==L_{1} ~||~L_{c}==L_{0}-L_{1}+1$}
\STATE $Flag=0$.
\ENDIF
\WHILE{$Flag$}
\STATE Running RF Chains Set Generation: Part \uppercase\expandafter{\romannumeral2}.
\STATE Running RF Chains Set Generation: Part \uppercase\expandafter{\romannumeral3}.
\ENDWHILE
\end{algorithmic}
\end{algorithm}

\setcounter{algorithm}{4}
\begin{algorithm}
\caption{RF Chains Set Generation: Part \uppercase\expandafter{\romannumeral2}}\label{TurningParaOptimization}
\begin{algorithmic}[1]
\IF{$L^{\lambda_{L}}-L^{\lambda_{U}}>1$}
\FOR{$l=1,\cdots,L^{\lambda_{L}}-L^{\lambda_{U}}-1$}
\STATE $\lambda=\lambda_{L}+\frac{\lambda_{U}-\lambda_{L}}{L^{\lambda_{L}}-L^{\lambda_{U}}}l$.
\IF{$\lambda$ is not in any interval of $\mathcal{I}$}  
\STATE Solve (\ref{HybridMultiUserBeam15}) with Algorithm~\ref{HybridMultiUserBeamA01} and $\lambda$, obtain $\bm{\Xi}^{\lambda}$, $L^{\lambda}$, and $\varrho^{\lambda}$.
\IF{$L^{\lambda}>L_{0}$}
\STATE $\lambda_{Temp}=\lambda$, $L_{Temp}=L^{\lambda}$.
\ELSE
\IF{$L^{\lambda}\notin\mathcal{L} $}
\STATE $\mathcal{L}=\mathcal{L} \cup \left\{\left(L^{\lambda},\lambda\right)\right\}$,  $L_{c}=L_{c}+1$.
\STATE Let $\mathcal{A}_{L^{\lambda}}$ be the set of the selected codewords indices obtained by $\bm{\Xi}^{\lambda}$, $\mathcal{A}
=\mathcal{A}\cup\left\{\mathcal{A}_{L^{\lambda}}\right\}$.
\ELSE
\STATE Sort the entries of $\mathcal{L}$ in ascending order with respect to $\lambda$.
\STATE Find the set of the indices $\mathcal{T}$ such that $\mathcal{T}=\left\{\iota : \mathcal{L}\left(\iota\right)\left(1,1\right)=L^{\lambda}\right\}$.
\IF{$\left|\mathcal{T}\right|==1$}
\STATE $\mathcal{I}=\mathcal{I}\cup\left\{\left[\min\left(\mathcal{L}\left(\mathcal{T}\left(1\right)\right)\left(1,2\right),\lambda\right),\atop
\max\left(\mathcal{L}\left(\mathcal{T}\left(1\right)\right)\left(1,2\right),\lambda\right)\right]\right\}$, $\mathcal{L}= \mathcal{L} \cup \left\{\left(L^{\lambda},\lambda\right)\right\}$.
\ELSIF{$\left|\mathcal{T}\right|==2$}
\STATE $\mathcal{I}=\mathcal{I}\setminus\left\{\left[\mathcal{L}\left(\mathcal{T}\left(1\right)\right)\left(1,2\right),
\mathcal{L}\left(\mathcal{T}\left(2\right)\right)\left(1,2\right)\right]\right\}$.
\STATE $\mathfrak{a}=\min\left(\mathcal{L}\left(\mathcal{T}\left(1\right)\right)\left(1,2\right),\lambda\right)$,
$\mathfrak{b}=\max\left(\mathcal{L}\left(\mathcal{T}\left(2\right)\right)\left(1,2\right),\lambda\right)$.
\STATE $\mathcal{I}=\mathcal{I}\cup\left\{\left[\mathfrak{a},\mathfrak{b}\right]\right\}$,
$\mathcal{L}\left(\mathcal{T}\left(1\right)\right)\left(1,2\right)=\mathfrak{a}$,
$\mathcal{L}\left(\mathcal{T}\left(2\right)\right)\left(1,2\right)=\mathfrak{b}$.
\ENDIF
\ENDIF
\ENDIF
\ENDIF
\ENDFOR
\ENDIF
\end{algorithmic}
\end{algorithm}

\setcounter{algorithm}{4}
\begin{algorithm}
\caption{RF Chains Set Generation: Part \uppercase\expandafter{\romannumeral3}}\label{TurningParaOptimization}
\begin{algorithmic}[1]
\IF{$L_{c}==L_{0}-L_{1}+1$}
\STATE $Flag=0$
\ELSE
\STATE Sort $\mathcal{L}$ in ascending order with respect to $\lambda$. 
\IF{$L_{0}>\mathcal{L}\left(1\right)\left(1,1\right)$}
\STATE $\lambda_{L}=\lambda_{Temp}$, $L^{\lambda_{L}}=L_{Temp}$,
\STATE $\lambda_{U}=\mathcal{L}\left(1\right)\left(1,2\right)$, $L^{\lambda_{U}}=\mathcal{L}\left(1\right)\left(1,1\right)$,
\ELSE
\STATE $Index=1$, $flag=1$;
\WHILE{$Index\leqslant \left|\mathcal{L}\right|-1 \&~flag$}
\IF{$\mathcal{L}\left(Index\right)\left(1,1\right)-\mathcal{L}\left(Index+1\right)\left(1,1\right)\geqslant 2$}
\STATE $\lambda_{L}=\mathcal{L}\left(Index\right)\left(1,2\right)$, $L^{\lambda_{L}}=\mathcal{L}\left(Index\right)\left(1,1\right)$.
\STATE $\lambda_{U}=\mathcal{L}\left(Index+1\right)\left(1,2\right)$.
\STATE $L^{\lambda_{U}}=\mathcal{L}\left(Index+1\right)\left(1,1\right)$, $flag=0$.
\ELSE
\STATE $Index=Index+1$,
\ENDIF
\ENDWHILE
\ENDIF
\ENDIF
\end{algorithmic}
\end{algorithm}

\subsection*{C. Refined Solution for EEmax problem}

Due to the introduction of the mixed $\ell_{1,\infty}$-norm squared for the selection of the RF chains in the previous subsections, the energy efficient beamforming vector cannot be directly extracted from the solution of~\eqref{HybridMultiUserBeam43}, i.e., $\left\{\bm{X}_{k}\right\}$. Therefore, we need to construct the reduced-size channel $\overline{\bm{h}}_{k}=\widehat{\bm{F}}^{H}\bm{h}_{k}$ according to the codewords selected by Algorithm~\ref{TurningParaOptimization}. Thus, the reduce-sized EEmax problem is given by
\begin{subequations}\label{HybridMultiUserBeam45}
\begin{align}
&\max_{\left\{\overline{\bm{g}}_{k}\right\}}~\frac{\sum\limits_{k=1}^{K} \overline{R}_{k}}
{\epsilon\sum\limits_{k=1}^{K}\left\|\widehat{\bm{F}}\overline{\bm{g}}_{k}\right\|_{2}^2+P_{dyn}^{*}},\label{HybridMultiUserBeam45a}\\
s.t.~&\overline{\mbox{SINR}}_{k}\geq \overline{\gamma}_{k}, \forall k\in\mathcal{K}, \sum\limits_{k=1}^{K}\left\|\widehat{\bm{F}}\overline{\bm{g}}_{k}\right\|_{2}^2\leq P,\label{HybridMultiUserBeam45b}
\end{align}
\end{subequations}
where $P_{dyn}^{\lambda_{L}^{key}}=L^{\lambda_{L}^{key}}\left(P_{RFC}+ M P_{PS}+P_{DAC}\right)+P_{sta}$. Problem~\eqref{HybridMultiUserBeam45} can be formulated as:
\begin{subequations}\label{HybridMultiUserBeam46}
\begin{align}
&\max_{\left\{\overline{\bm{g}}_{k},\overline{\alpha}_{k}, \overline{\beta}_{k}, \overline{\phi}_{k}\right\},\overline{\tau}, \overline{\chi}, \overline{\kappa}}~\overline{\chi},\label{HybridMultiUserBeam46a}\\
s.t.~&\sum\limits_{k=1}^{K}\overline{\beta}_{k}\geqslant\overline{\tau}^{2}, \frac{\overline{\tau}^{2}}{\overline{\kappa}}\geqslant\overline{\chi},\label{HybridMultiUserBeam46b}\\
&\epsilon\sum\limits_{k=1}^{K}\left\|\widehat{\bm{F}}\overline{\bm{g}}_{k}\right\|_{2}^2+P_{dyn}^{*}
\leqslant\overline{\kappa},\label{HybridMultiUserBeam46c}\\
&\eqref{HybridMultiUserBeam28b}, \eqref{HybridMultiUserBeam28c}, \eqref{HybridMultiUserBeam28d}.\label{HybridMultiUserBeam46d}
\end{align}
\end{subequations}
Similarly, instead of directly solving~(\ref{HybridMultiUserBeam46}), we resort to solving the following convex approximated problem
\begin{subequations}\label{HybridMultiUserBeam47}
\begin{align}
&\max_{\left\{\overline{\bm{g}}_{k},\overline{\alpha}_{k}, \overline{\beta}_{k}, \overline{\phi}_{k}\right\},\overline{\tau}, \overline{\chi}, \overline{\kappa}}~\overline{\chi},\label{HybridMultiUserBeam47a}\\
s.t.~&\sum\limits_{k=1}^{K}\overline{\beta}_{k}\geqslant\overline{\tau}^{2},
\overline{\Psi}^{\left(I\right)}\left(\overline{\tau},\overline{\kappa}\right)\geqslant\overline{\chi},
\eqref{HybridMultiUserBeam46c}, \eqref{HybridMultiUserBeam46d},\label{HybridMultiUserBeam47b}
\end{align}
\end{subequations}
where $I$ denotes the $I$th iteration, and $\overline{\Psi}^{\left(I\right)}\left(\overline{\tau},\overline{\kappa}\right)$ is given by
\begin{equation}\label{HybridMultiUserBeam48}
\overline{\Psi}^{\left(I\right)}\left(\overline{\tau},\overline{\kappa}\right)\triangleq
2\frac{\overline{\tau}^{\left(I\right)}}{\overline{\kappa}^{\left(I\right)}}\overline{\tau}
-\left(\frac{\overline{\tau}^{\left(I\right)}}{\overline{\kappa}^{\left(I\right)}}\right)^{2}\overline{\kappa}
\end{equation}
Thus, problem~(\ref{HybridMultiUserBeam47}) can be solved in a similar manner as described in Algorithm~\ref{HybridMultiUserBeamA03}.

\section*{\sc \uppercase\expandafter{\romannumeral6}. Numerical Results}

In this section, we present numerical results to demonstrate the performance of our developed RF-baseband hybrid precoding design. A uniform linear array with antenna spacing equal to a half wavelength is adopted, and the RF phase shifters use quantized phases. The predesigned codebook $\bm{F}$ is the DFT codebook. The propagation environment is modeled as $N_{cl}=6$ with $N_{ray}=8$ for each cluster with Laplacian distributed angles of departure. For simplicity, we assume that all clusters are of equal power, i.e., $\sigma_{\alpha,m_{p}}^{2}=\sigma_{\alpha}^{2}, \forall m_{p}$~\cite{TWCEl2014}. The mean cluster angle of $\phi_{m_{p}}$ is uniformly distributed over $\left[-\pi, \pi\right)$, and the constant angular spread of AoD $\sigma_{\phi}$ is $7.5^{o}$. $P_{RFC}=43$ mW, $P_{PA}= 20$ mW, $P_{DAC}= 200$ mW, $P_{PS} = 30$ mW, $P_{mixer}= 19$ mW, $P_{BB}=300$ mW, and $P_{cool}=200$ mW~\cite{IJSSCYu2010,TMTTLi2013}. The noise power spectrum density is $\sigma^{2}=1$. For fairness, all simulated precoding designs use the same total power constraint and the signal-to-noise ratio is defined as $\rm{SNR}=10\log_{10}\left(\frac{P}{\sigma^{2}}\right)$. The inefficiency factor of power amplifier $\epsilon$ is set to unit and the stop threshold  is $\zeta=10^{-3}$.
In all simulation figures, the simulated EE of the system is given by
\begin{equation}\label{HybridMultiUserBeam49}
\frac{\sum\limits_{k=1}^{K} \log_{2}\left(1+\frac{\left\|\bm{h}_{k}^{H}\widehat{\bm{F}}\overline{\bm{g}}_{k}\right\|_{2}^{2}}{\sum\limits_{l=1,l\neq k}^{K}\left\|\bm{h}_{k}^{H}\widehat{\bm{F}}\overline{\bm{g}}_{l}\right\|_{2}^{2}+\sigma^{2}}\right)}
{\epsilon\sum\limits_{k=1}^{K}\left\|\widehat{\bm{F}}\overline{\bm{g}}_{k}\right\|_{2}^2+P_{dyn}^{*}}.
\end{equation}
We compare the performance of the proposed strategy to the optimal fully digital precoder with one RF chain per antenna, whose EE is calculated as
\begin{equation}\label{HybridMultiUserBeam50}
\frac{\sum\limits_{k=1}^{K} \log_{2}\left(1+\frac{\left\|\bm{h}_{k}^{H}\overline{\bm{g}}_{k}\right\|_{2}^{2}}{\sum\limits_{l=1,l\neq k}^{K}\left\|\bm{h}_{k}^{H}\overline{\bm{g}}_{l}\right\|_{2}^{2}+\sigma^{2}}\right)}
{\epsilon\sum\limits_{k=1}^{K}\left\|\overline{\bm{g}}_{k}\right\|_{2}^2+M\left(P_{RFC}+P_{DAC}+P_{PA}\right)+P_{BB}+P_{cool}}.
\end{equation}

In our simulation scenario, fully digital precoding denotes using Algorithm~3 to solve the SRmax problem, where each antenna connects with an independent RF channel at the BS. Fully analog beamforming is achieved by selecting the best codeword from the codebook via beam training and setting the baseband precoder as an identity matrix with uniform power allocation between users. OMP SRmax Hybrid Precoding uses the orthogonal matching pursuit method~\cite{TWCEl2014} to obtain the RF-baseband precoders based on the solution of the fully digital SRmax problem.

Fig.~\ref{HybridSumRateComparisons} and Fig.~\ref{HybridSumRateComparisonsOMP} show the SR performance of various hybrid precoding designs as well as the fully digital precoder. The results are obtained by averaging over $1000$ random channel realizations. The target rate of the $k$th user is set to be zero. Numerical results show that the proposed hybrid precoding design achieves the highest SR among several hybrid RF-baseband precoders. This is because the proposed hybrid RF-baseband precoding design provides a more flexible way to achieve the beam diversity gain. One can see that the fully analog beamforming method results in the worst SR performance, indicating that the inter-user interference cannot be effectively suppressed.
\begin{figure}[h]
\centering
\captionstyle{flushleft}
\onelinecaptionstrue
\includegraphics[width=0.8\columnwidth,keepaspectratio]{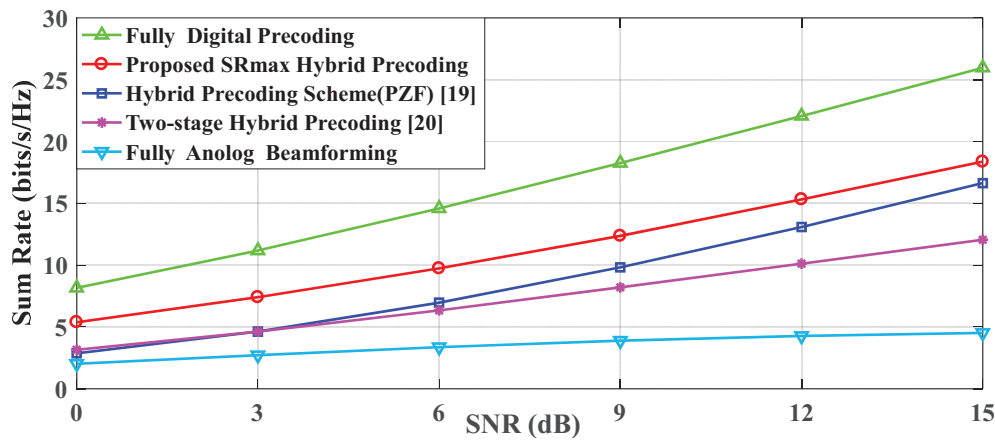}\\
\caption{Sum rate comparison of various hybrid precoders and fully digital precoder, $M=N=16$, $S=4$, $K=4$.}
\label{HybridSumRateComparisons}
\end{figure}

\begin{figure}[th]
\centering
\subfigure[$M=N=16$, $S=4$, $K=2$]{
\label{HybridSumRateComparisonsOMPa} 
\includegraphics[width=0.8\columnwidth]{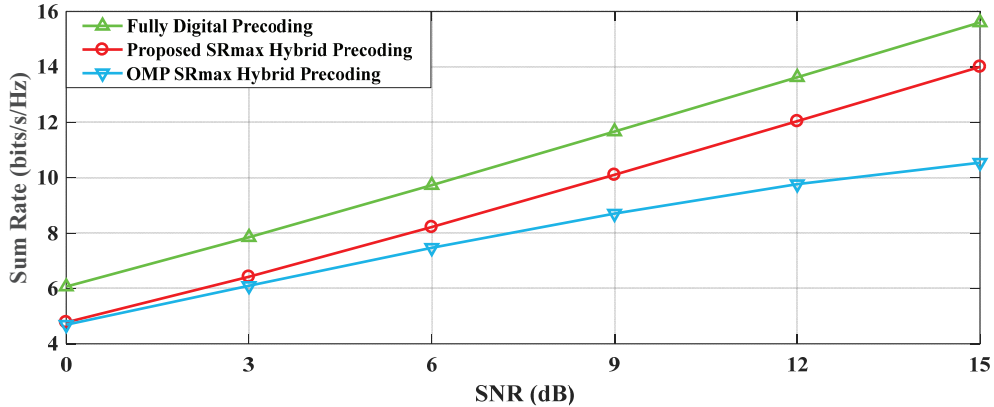}}
\subfigure[$M=N=32$, $S=4$, $K=2$]{
\label{HybridSumRateComparisonsOMPb} 
\includegraphics[width=0.8\columnwidth]{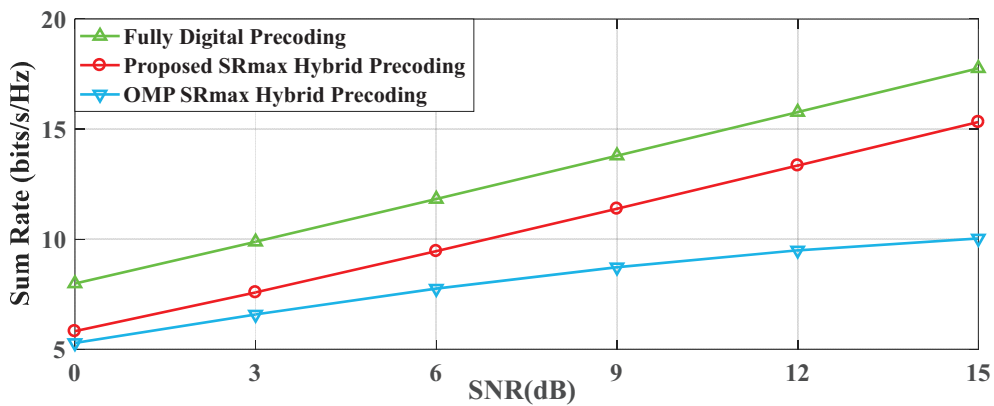}}
\caption{Sum rate comparison of various hybrid precoders and fully digital precoder.}
\label{HybridSumRateComparisonsOMP}
\end{figure}

Fig.~\ref{Chi&LambdaComparison} illustrates the change of the objective $\chi$ of \eqref{HybridMultiUserBeam40} versus an increasing $\lambda$ for two random channel realizations. The target rate of the $k$th user is set to be the rate achieved by randomly selecting $S$ analog codeword from codebook $\bm{F}$ and using the baseband precoder as $\underline{\bm{G}}=\frac{P}{K}\upsilon_{max}^{\left(K\right)}\left(\bm{H}\underline{\bm{F}}\right)$. Simulation results show that there exist indeed breaking points of $\lambda$, i.e., the number of selected RF chains keep unchanged within a range of $\lambda$ but suddenly changes at some points. One can observe that the number of selected RF chains decreases with an increasing value of $\lambda$. Within a certain interval of $\lambda$, the number of selected RF chains keep the same but the objective $\chi$ of \eqref{HybridMultiUserBeam40} decreases when $\lambda$ increases. Table~\ref{Chi&LambdaComparisonList} lists the set of the indices of the selected RF chains corresponding to the channel realization used in Fig.~\ref{Chi&LambdaComparison}. One can observe that the same set of the RF chains (or codewords) are selected for any $\lambda\in\left[\lambda_{L+1}^{key},\lambda_{L}^{key}\right)$. This observation is consistent with the result in Theorem~\ref{HybridMultiUserBeamL03}, which has been used in Algorithm~\ref{TurningParaOptimization}.
\begin{figure}[th]
\centering
\subfigure[Channel Realization I]{
\label{ChiLambdaComparisona} 
\includegraphics[width=0.8\columnwidth]{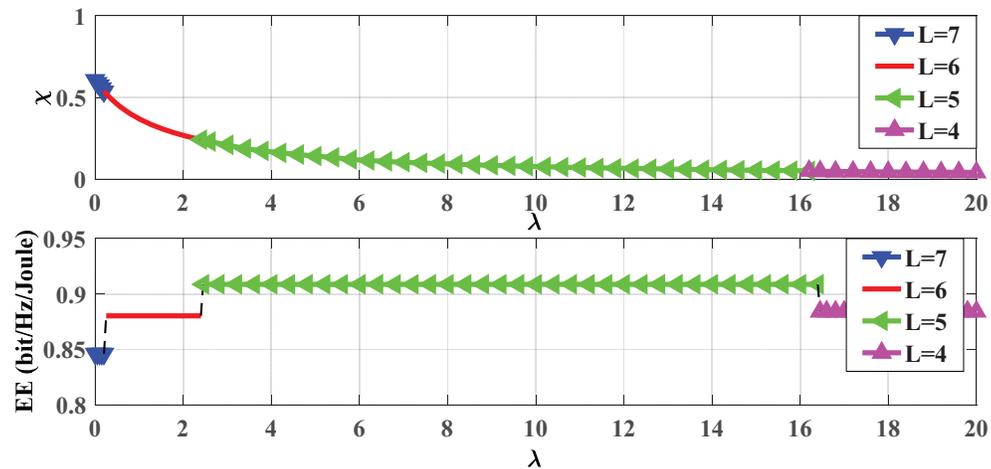}}
\subfigure[Channel Realization II]{
\label{ChiLambdaComparisonb} 
\includegraphics[width=0.8\columnwidth]{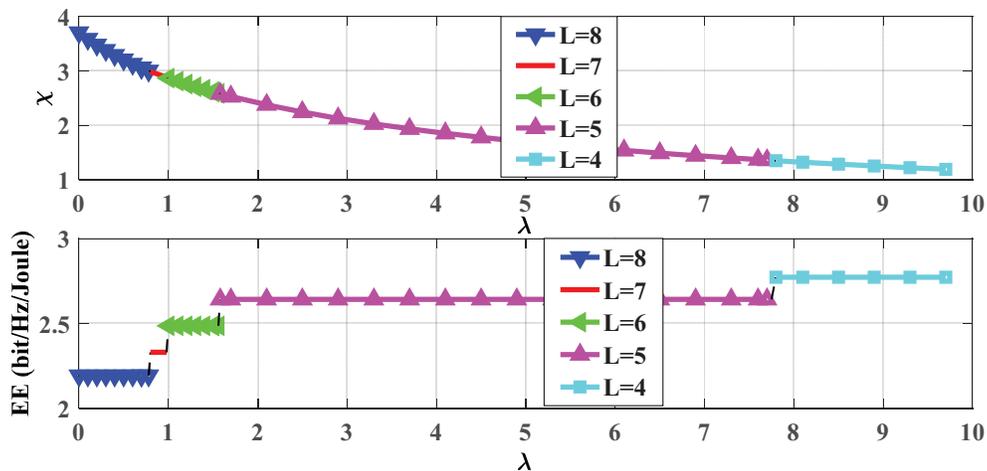}}
\caption{$\chi$ versus value of $\lambda$, $M=N=16$, $S=8$, $K=4$, $\mbox{SNR}=15$dB.}
\label{Chi&LambdaComparison} 
\end{figure}

\begin{table*}[htbp]
\setlength{\abovecaptionskip}{0pt}
\setlength{\belowcaptionskip}{5pt}
\captionstyle{flushleft}
\onelinecaptionstrue
\centering
\caption{Examples of $\lambda$, $\mathcal{A}_{L}$, and $L$ with $M=N=16$, $S=8$, $K=4$, $\mbox{SNR}=15$dB.}
\begin{tabular}{|c|c|c|c|c|c|c|c|c|}
\hline
$\lambda$ &$\mathcal{A}_{L}$&$L$& EE (bits/Hz/Joule) &&$\lambda$&$\mathcal{A}_{L}$&$L$&EE (bits/Hz/Joule)\\
\hline
$0$& $3,4,5,6,7,8,16$& $7$&$0.8460$&&$2.45$&$4,5,6,8,16$&$5$&$0.9087$\\
\hline
$0.2$& $3,4,5,6,7,8,16$& $7$&$0.8460$&&$16.35$&$4,5,6,8,16$&$5$&$0.9087$ \\
\hline
$.25$& $4,5,6,7,8,16$& $6$&$0.8803$&&$16.4$&$4,5,6,16$&$4$&$0.8841$\\
\hline
$2.4$& $4,5,6,7,8,16$& $6$&$0.8803$&&$20$&$4,5,6,16$&$4$& $0.8841$\\
\hline
\hline
\hline
$\lambda$ & $\mathcal{A}_{L}$ &$L$&EE  (bits/Hz/Joule)&&$\lambda$& $\mathcal{A}_{L}$&$L$&EE (bits/Hz/Joule)\\
\hline
$0$& $1,2,3,9,10,14,15,16$& $8$&$2.1931$&&$1.56$&$1,2,3,10,15,16$&$6$&$2.4861$\\
\hline
$0.4$& $1,2,3,9,10,14,15,16$& $8$&$2.1931$&&$1.58$&$1,2,3,10,15$&$5$&$2.6415$ \\
\hline
$0.78$& $1,2,3,9,10,14,15,16$& $8$&$2.1931$&&$5.9$&$1,2,3,10,15$&$5$&$2.6415$\\
\hline
$0.8$& $1,2,3,10,14,15,16$& $7$&$2.3299$&&$7.75$&$1,2,3,10,15$&$5$& $2.6415$\\
\hline
$0.9$& $1,2,3,10,14,15,16$& $7$&$2.3299$&&$7.8$&$1,3,10,15$&$4$&$2.7720$\\
\hline
$0.98$& $1,2,3,10,14,15,16$& $7$&$2.3299$&&$8.8$&$1,3,10,15$&$4$&$2.7720$ \\
\hline
$1$& $1,2,3,10,15,16$& $6$&$2.4861$&&$9.8$&$1,3,10,15$&$4$&$2.7720$\\
\hline
$1.3$& $1,2,3,10,15,16$& $6$&$2.4861$&&$-$&$-$&$-$&$-$ \\
\hline
\end{tabular}
\label{Chi&LambdaComparisonList}
\end{table*}

Fig.~\ref{HybridDigitalEESRComparison32a} and Fig.~\ref{HybridDigitalEESRComparison32b} illustrates the SR and EE of various mmWave precoding designs and the fully digital precoding design, respectively. The results are obtained by averaging over $1000$ random channel realizations. The target rate of $k$th user is set to be the rate achieved by selecting randomly $S$ analog codeword from codebook $\bm{F}$ and defining the baseband precoder as $\underline{\bm{G}}=\frac{P}{K}\upsilon_{max}^{\left(K\right)}\left(\bm{H}\underline{\bm{F}}\right)$. It is observed that using the DFT codebook in the proposed EEmax precoder is better than using the 802.15.3c codebook in terms of the EE performance, while the two codebooks lead to the similar SR performance. Compared to the fully digital precoder, the proposed SRmax/EEmax hybrid precoders have certain system SR performance loss as the RF-baseband hybrid architecture may not fully exploit the multi-path diversity gain. The circuit power consumption of the hybrid architecture increases with the number of phase shifter and the number of mixers, which are determined by the number of transmit antennas and the number of RF chains. Therefore, the system EE performance of the hybrid precoding is also determined by the number of transmit antennas and the number of RF chains. One can see that different configuration of the number of antennas, RF chains, and users leads to different system EE performance. For example, for the configuration of $M=N=32$, $S=4$, and $K=2$, the EE performance of the hybrid precoders is better than that of the fully digital precoder. Besides, the fully digital precoder leads to a much higher hardware cost ($M=32$ RF chains versus $S=8$ RF chains), which is critical for mmWave communication systems using GHz bandwidth and even higher sampling rates.

\begin{figure}[h]
\centering
\captionstyle{flushleft}
\onelinecaptionstrue
\includegraphics[width=0.8\columnwidth,keepaspectratio]{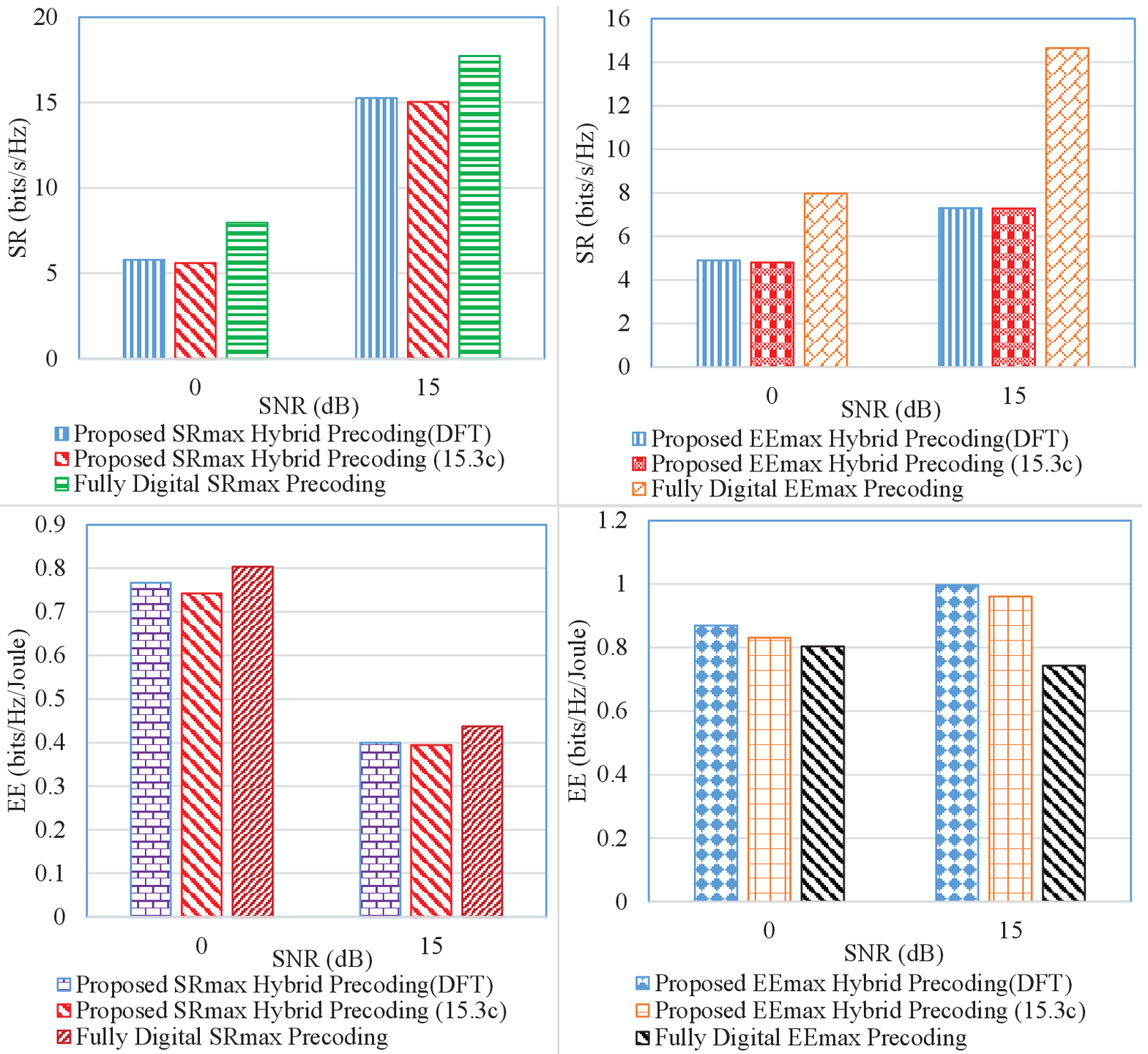}\\
\caption{EE/SR comparison of the hybrid precoders and fully digital precoder, $M=N=32$, $S=4$, $K=2$.}
\label{HybridDigitalEESRComparison32a}
\end{figure}

\begin{figure}[h]
\centering
\captionstyle{flushleft}
\onelinecaptionstrue
\includegraphics[width=0.8\columnwidth,keepaspectratio]{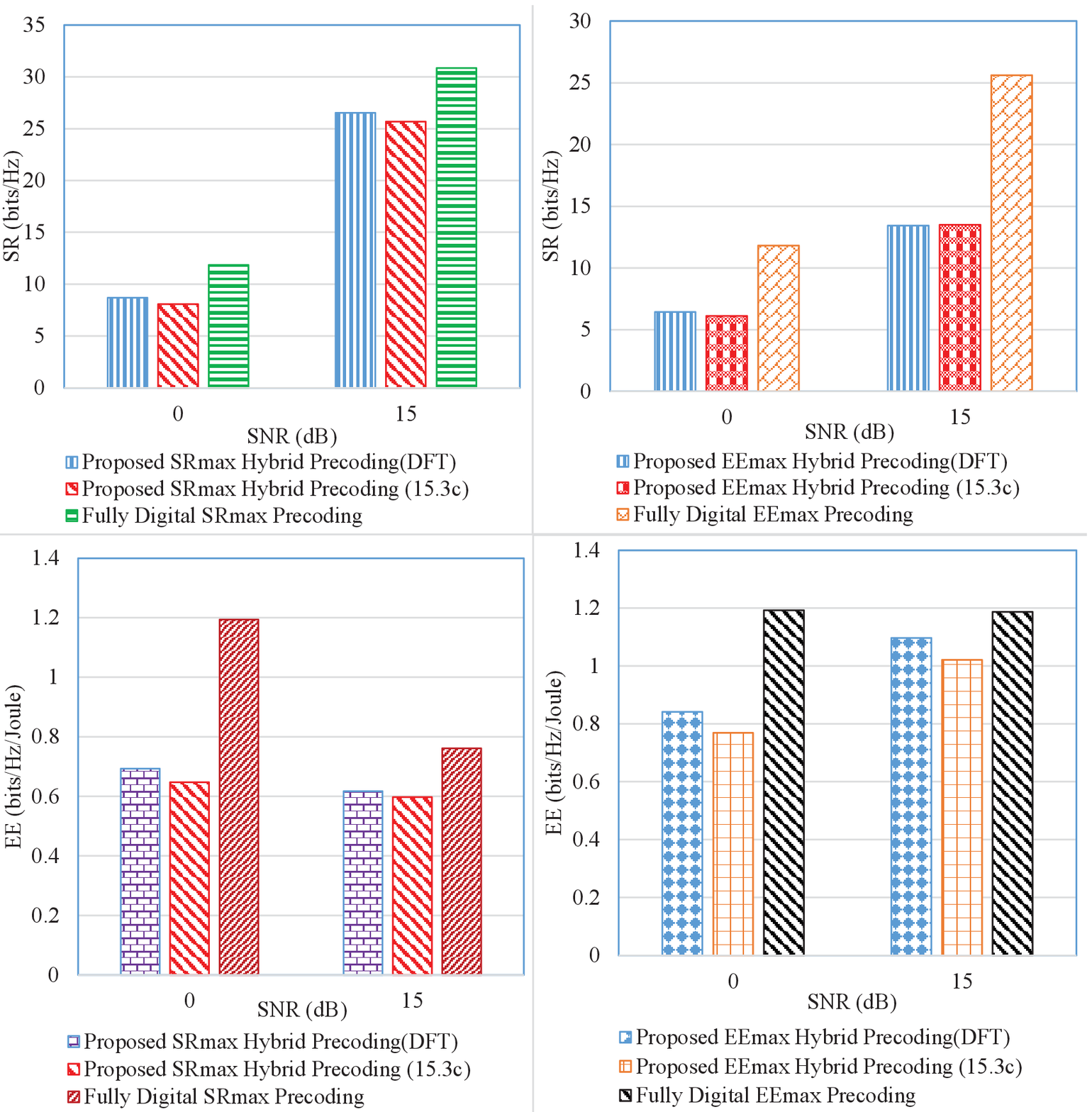}\\
\caption{EE/SR comparison of the hybrid precoders and fully digital precoder, $M=N=32$, $S=8$, $K=4$.}
\label{HybridDigitalEESRComparison32b}
\end{figure}

\section*{\sc \uppercase\expandafter{\romannumeral7}. Conclusions}

In this paper, we considered the design of the hybrid RF-baseband precoding for the downlink of multiuser multi-antenna systems with the aim to maximize the system SR and the system EE. We developed a codebook based RF precoding method and obtained the channel state information via a beam sweep procedure. Exploiting the codebook based design, we simplified the complicated hybrid precoders optimization problems to JWSPD problems. Then, efficient methods were developed to address the JWSPD problems for maximizing the SR and EE of the system. Finally, extensive numerical simulation results are provided to validate the effectiveness of the proposed hybrid precoding design.
%

\begin{small}

\end{small}


\begin{thebibliography}{30}
\bibitem{CommChen2012} S. Chen, Y. Wang, W. Ma, and J. Chen, ``Technical innovations promoting standard evolution: From TD-SCDMA to TD-LTE and beyond,'' \emph{IEEE Wireless Commun.}, vol. 19, no. 2, pp. 60-66, Feb. 2012.
\bibitem{TITBjor2014}E. Bj\"{o}rnson, J. Hoydis, M. Kountouris, and M. Debbah, ``Massive MIMO systems with non-ideal hardware: EE, estimation and capacity limits," \emph{IEEE Trans. Inf. Theory}, vol. 60, no. 11, pp. 7112-7139, Nov. 2014.
\bibitem{CMagTao2012}X. Tao, X. Xu, and Q. Cu, ``An overview of cooperative communications," \emph{IEEE Commun. Mag.}, pp. 65-71, Jun. 2012.
\bibitem{JSACWang2007}J. Wang, W. Guan, Y. Huang, R. Schober, and X. You, ``Distributed optimization of hierarchical small cell networks: A GNEP framework," \emph{IEEE J. Sel. Areas Commun.}, vol. 35, no. 2, pp. 249-264, Feb. 2017.
\bibitem{JSTSPHeath2016}R. Heath Jr., N. Gonzalez-Prelcic, S. Rangan, W. Roh, and A. Sayeed, ``An overview of signal processing techniques for millimeter wave MIMO systems," \emph{IEEE J. Sel. Topics Signal. Process.}, vol. 10, no. 4, pp. 436-453, Apr. 2016.
\bibitem{ProcRangan2014}S. Rangan, T. S. Rappaport, and E. Erkip, ``Millimeter-wave cellular wireless networks: Potentials and challenges," \emph{Proc. IEEE,} vol. 102, no. 3, pp. 366-385, Mar. 2014.
\bibitem{AccessRap2013}T. Rappaport, S. Sun, R. Mayzus, H. Zhao, Y. Azar, \emph{et al.}, ``Millimeter wave mobile communications for 5G cellular: It will work!" \emph{IEEE Access}, vol. 1, pp. 335-349, 2013.
\bibitem{TAPRap2013}T. Rappaport, F. Gutierrez, E. Ben-Dor, J. Murdock, Y. Qiao, and J. Tamir, ``Broadband millimeter-wave propagation measurements and models using adaptive-beam antennas for outdoor urban cellular communications," \emph{IEEE Trans. Antennas Propag.}, vol. 61, no. 4, pp. 1850-1859, Apr. 2013
\bibitem{AccessMac2015}G. MacCartney Jr. \emph{et al.}, ``Indoor office wideband millimeter-wave propagation measurements and models at 28 GHz and 73 GHz for ultra-dense 5G wireless networks (Invited Paper)," \emph{IEEE Access}, vol, 3, pp. 2388-2424 2015.
\bibitem{JSACCui2004}S. Cui, A. J. Goldsmith, and A. Bahai, ``Energy-efficiency of MIMO and cooperative MIMO techniques in sensor networks," \emph{IEEE J. Sel. Areas Commun.}, vol. 22, no. 6, pp. 1089-1098, Aug. 2004.
\bibitem{TSPZhang2005}X. Zhang, A. Molisch, and S. Kung, ``Variable-phase-shift-based RF-baseband codesign for MIMO antenna selection," \emph{IEEE Trans. Sig. Proc.}, vol. 53, no. 11, pp. 4091-4103, Nov. 2005.
\bibitem{JSACWang2009}J. Wang et al., ``Beam codebook based beamforming protocol for multi Gbps millimeter-wave WPAN systems," \emph{IEEE J. Sel. Areas Commun.}, vol. 27, no. 8, pp. 1390-1399, Aug. 2009.
\bibitem{CSTKutty2016}S. Kutty, and D. Sen, ``Beamforming for millimeter wave communications: An inclusive survey," \emph{IEEE Commun. Suerveys \& Tutorials}, vol. 18, no. 2, pp. 949-973, Second Quarter 2016.
\bibitem{TCOMHur2013}S. Hur, T. Kim, D. Love, J. Krogmeier, T. Thomas, and A. Ghosh, ``Millimeter wave beamforming for wireless backhaul and access in small cell networks," \emph{IEEE Trans. Commun.}, vol. 61, no. 10, pp. 4391-4403, Oct. 2013.

\bibitem{TSPVen2010}V. Venkateswaran and A. van der Veen, ``Analog beamforming in MIMO communications with phase shift networks and online channel estimation," \emph{IEEE Trans. Sig. Proc.}, vol. 58, no. 8, pp. 4131-4143, Aug. 2010.
\bibitem{TWCEl2014}O. El Ayach, S. Rajagopal, S. Abu-Surra, Z. Pi, and R. Heath, ``Spatially sparse precoding in millimeter wave MIMO systems," \emph{IEEE Trans. Wireless Commun.}, vol. 13, no. 3, pp. 1499-1513, Mar. 2014.
\bibitem{JSTSPAlk2014}A. Alkhateeb, O. El Ayach, G. Leus, and R. Heath, ``Channel estimation and hybrid precoding for millimeter wave cellular systems," \emph{IEEE J. Sel. Topics Signal. Process.}, vol. 8, no. 5, pp. 831-846, Oct. 2014.
\bibitem{TVTXue2016}Q. Xue, X. Fang, M. Xiao, and Y. Li, ``Multi-user millimeter wave communications with nonorthogonal beams," \emph{IEEE Trans. on Veh. Tech., Special issue on Tera-Hz Commun.,}  Oct. 2016.
\bibitem{CLLiang2014}L. Liang, W. Xu, and X. Dong, ``Low-complexity hybrid precoding in massive multiuser MIMO systems," \emph{IEEE Commun. Letter}, vol. 3, no. 6, pp. 653-656, Dec. 2014.
\bibitem{TWCAlkha2015}A. Alkhateeb, G. Leus, and R. Heath, Jr., ``Limited feedback hybrid precoding for multi-user millimeter wave systems," \emph{IEEE Trans. Wireless Commun.}, vol. 14, no. 11, pp. 6481-6494, Nov. 2015.
\bibitem{TCOMHE2013}S. He, Y. Huang, S. Jin, and L. Yang, ``Coordinated beamforming for energy efficient transmission in multicell multiuser systems, '' \emph{IEEE Trans. Commun.}, vol. 61, no. 12, pp. 4961-4971, Dec. 2013.
\bibitem{CLNguyen2015}K. Nguyen, L. Tran, O. Tervo, Q. Vu, and M. Juntti, ``Achieving energy efficiency fairness in multicell MISO downlink," \emph{IEEE Commun. Lett.,} vol. 19. no. 8, pp. 1426-1429, Aug. 2015.
\bibitem{TWCVen2015}L. Venturino, A. Zappone, C. Risi, and S. Buzzi, ``Energy-efficient scheduling and power allocation in downlink OFDMA networks with base station coordination,'' \emph{IEEE Trans. Wireless Commun.}, vol. 14, no.1, pp. 1-14, 2015.
\bibitem{TSPZap2016}A. Zappone, L. Sanguinetti, G. Bacci, E. Jorswieck, and M. Debbah, ``Energy-efficient power control: A Look at 5G wireless technologies," \emph{IEEE Trans. on Sig. Proc.,} vol. 64, no. 7, pp. 1668-1683, Apr. 2016.
\bibitem{TCOMLI2014}H. Li, L. Song, and M. Debbah, ``Energy efficiency of large-scale multiple antenna systems with transmit antenna selection," \emph{IEEE Trans. Commun.}, vol. 62, no. 2, pp. 638-647, Feb. 2014.
\bibitem{AccessRial2016}R. Rial, C. Rusu, N.Prelcic, A. Alkhateeb, and R. Heath, Jr, ``Hybrid MIMO architectures for millimeter wave communications: Phase shifters or switches?"  \emph{IEEE Access}, vol, 4, pp. 247-267, 2016.
\bibitem{JSACZi2016}R. Zi, X. Ge, J. Thompson, C. Wang, H. Wang, and T. Han, ``Energy efficiency optimization of 5G radio frequency chains systems," \emph{IEEE J. Sel. Areas Commun.}, vol. 34, no. 4, pp. 758-771, Apr. 2016.
\bibitem{IEEE802153c}``Part 15.3: Wireless MAC and PHY specifications for high rate WPANs. Amendment 2: Millimeter-wave-based alternative physical layer extension,'' \emph{IEEE Std 802.15.3c-2009}, Oct. 2009.
\bibitem{IEEE802.11ad}``Part 11: Wireless LAN medium access control (MAC) and physical layer (PHY) specifications. Amendment 3: Enhancements for very high throughput in the 60 GHz band'' \emph{IEEE Std. 802.11ad-2012}, Dec. 2012.
\bibitem{TWCZhang2007}M. Zhang, P. Smith, and M. Shafi, ``An extended one-ring MIMO channel model," \emph{IEEE Trans. Wireless Commun.}, vol. 6, no. 8, pp. 2759-2764, Aug. 2007.
\bibitem{AntennaTheory1997}C. Balanis, Antenna Theory. Wiley, 1997.
\bibitem{TWCKwanNg2012}D. Kwan Ng, E. Lo, and R. Schober, ``Energy-efficient resource allocation in multi-cell OFDMA systems with limited backhaul capacity," \emph{IEEE Trans. Wireless Commun.}, vol. 11, no. 10, pp. 3618-3631, Sep. 2012.
\bibitem{ConfZhou2012}L. Zhou and Y. Ohashi, ``Efficient codebook-based MIMO beamforming for millimeter-wave WLANs," in Proc. \emph{IEEE 23rd Int. Symp. Pers. Indoor Mobile Radio Commun. (PIMRC)}, pp. 1885-1889, 2012.
\bibitem{BoolBoyd2004}S. Boyd and L. Vandenberghe, Convex Optimization. Cambridge, U.K.: Cambridge Univ. Press, 2004.
\bibitem{TSPMehanna2013} O. Mehanna, N. Sidiropoulos, and G. Giannakis, ``Joint multicast beamforming and antenna selection," \emph{IEEE Trans. Sig. Proc.}, vol. 61, no. 10, pp. 2660-2674, May. 2013.
\bibitem{BookYe1997}Y. Ye, Interior Point Algorithms: Theory and Analysis. New York: Wiley, 1997.
\bibitem{OperalMarks1977}B. Marks and G.  Wright, ``A general inner approximation algorithm for nonconvex mathematical programs," \emph{Operat. Res.}, vol. 26, no. 4, pp. 681-683, Jul./Aug. 1978.
\bibitem{GlagBibby1974}J. Bibby, ``Axiomatisations of the average and a further generalisation of monotonic sequences," \emph{Glasgow Math. J.}, vol. 15, pp. 63-65, 1974.
\bibitem{OperMarks1978}B. Marks and G. Wright, ``A general inner approximation algorithm for nonconvex mathematical programs," \emph{Oper. Res.}, vol. 26, no. 4, pp. 681-683, Jul./Aug. 1978.
\bibitem{TSPHan2013}M. Hanif, L. Tran, A. T\"{o}lli, M. Juntti, and S. Glisic, ``Efficient solutions for weighted sum rate maximization in multicellular networks  with channel uncertainties," \emph{IEEE Trans. Sig. Proc.}, vol. 61,  no. 22, pp. 5659-5674, Nov. 2013.
\bibitem{TSPSid2006}N. Sidiropoulos, T.  Davidson, and Z. Luo, ``Transmit beamforming for physical-layer multicasting," \emph{Trans. Sig. Proc.,} vol. 54, no. 6, pp. 2239-2251, 2006.
\bibitem{TSPKari2008}E. Karipidis, N. Sidiropoulos, and Z. Luo, ``Quality of service and max-min fair transmit beamforming to multiple co-channel multicast groups,'' \emph{IEEE Trans. Sig. Proc.,} vol. 56, no. 3, pp. 1268-1279, Mar. 2008.
\bibitem{TVTSchubert2004}M. Schubert and H. Boche, ``Solution of the multiuser downlink beamforming problem with individual SINR constraints," \emph{IEEE Trans. Veh. Technol.}, vol. 53, no. 1, pp. 18-28, Jan. 2004.
\bibitem{TITZhang2012}L. Zhang, R. Zhang, Y. Liang, Y. Xin, and H. Poor, ``On Gaussian MIMO BC-MAC duality with multiple transmit covariance constraints," \emph{IEEE Trans. Inf. Theory}, vol. 58, no. 4, pp. 2064-2078, Apr. 2012.
\bibitem{TCOMHe2015}S. He, Y. Huang, L. Yang, B. Ottersten, and W. Hong, ``Energy efficient coordinated beamforming for multicell system: Duality based algorithm design and massive MIMO transition," \emph{IEEE Trans. on Commun.}, vol. 63, no. 12, pp. 4920-4935, Dec., 2015.
\bibitem{TWCSadek2007}M. Sadek, A. Tarighat, and A. Sayed, ``A leakage-based precoding scheme for downlink multi-user MIMO channels," \emph{IEEE Trans. Wireless Commun.} vol. 6, no. 5, pp. 1711-1721, May 2007.
\bibitem{TSPWiessel2006}A. Wiesel, Y. C. Eldar, and S. Shamai, ``Linear precoding via conic optimization for fixed MIMO receivers,"  \emph{IEEE Trans. Sig. Proc.}, vol. 54, no. 1, pp. 161-176, Jan. 2006.
\bibitem{IJSSCYu2010}Y. Yu, P. Baltus, A. De Graauw, E. Van Der Heijden, C. Vaucher, and A. Van Roermund, ``A 60 GHz phase shifter integrated with LNA and PA in 65 nm CMOS for phased array systems," \emph{IEEE Journal of Solid-State Circuits}, vol. 45, no. 9, pp. 1697-1709, Sep. 2010.
\bibitem{TMTTLi2013}W. Li, Y. Chiang, J. Tsai, H. Yang, \emph{et al}, ``60-GHz 5-bit phase shifter with integrated VGA phase-error compensation," \emph{IEEE Trans. on Micro. Theory and Tech.}, vol. 61, no. 3, pp. 1224-1235, Mar. 2013.
\end{thebibliography}
\end{document}